\newif\iflong
\longtrue

\documentclass{article}
\usepackage[utf8]{inputenc}
\usepackage[T1]{fontenc}
\usepackage{fullpage}
\usepackage{amsfonts,amsmath,amssymb,amsthm}
\usepackage{cite,comment,enumerate,hyperref}
\usepackage[colorinlistoftodos,textsize=small,color=red!25!white,obeyFinal]{todonotes}
\usepackage{csquotes}
\usepackage{bm}
\usepackage[noend, boxed]{algorithm2e}
\SetKwComment{Comment}{/* }{ */}
\usepackage{makecell}

\usepackage{thmtools}
\usepackage{thm-restate}

\usepackage{array}

\newtheorem{theorem}{Theorem}[section]
\newtheorem{corollary}[theorem]{Corollary}

\newtheorem{lemma}[theorem]{Lemma}

\newtheorem{proposition}[theorem]{Proposition}
\theoremstyle{definition}
\newtheorem{definition}{Definition}
\theoremstyle{remark}	

\usepackage{xcolor,xspace}

\newcommand{\poly}{\texttt{poly}}

\newcommand{\quotes}[1]{``#1''}

\DeclareMathOperator*{\E}{\mathbb{E}}

\usepackage{mathtools}

\DeclarePairedDelimiter\floor{\lfloor}{\rfloor}

\newcommand{\LabelCover}{\textsc{Label Cover}}

\def\al    {\alpha}
\def\be    {\beta}

\def\opt {{\sf OPT}}

\def\hopset {{\sc Generalized $\beta$-Hopset}}
\def\jt {{\sc Min Density $(i,j)$-Junction Tree}}
\def\ljt {{\sc Min Density Length-Bounded Junction Tree}}

\title{Approximation Algorithms for Optimal Hopsets}
\author{Michael Dinitz\thanks{Johns Hopkins University.  Supported in part by NSF award 2228995.} \and Ama Koranteng\footnotemark[1] \and Yasamin Nazari\thanks{Vrije Universiteit Amsterdam.  This publication is part of the project VI.Veni.232.038 of the research programme ``Dynamic graph algorithms: distances and clustering'' which is financed by the Dutch Research Council (NWO).}}
\date{}

\begin{document}

\maketitle

\begin{abstract}
  For a given graph $G$, a \emph{hopset} $H$ with hopbound $\beta$ and stretch $\alpha$ is a set of edges such that between every pair of vertices $u$ and $v$, there is a path with at most $\beta$ hops in $G \cup H$ that approximates the distance between $u$ and $v$ up to a multiplicative stretch of $\alpha$. Hopsets have found a wide range of applications for distance-based problems in various computational models since the 90s. More recently, there has been significant interest in understanding these fundamental objects from an existential and structural perspective. 
  But all of this work takes a worst-case (or existential) point of view: How many edges do we need to add to satisfy a given hopbound and stretch requirement for \emph{any input graph}?
  
  We initiate the study of the natural optimization variant of this problem: given a \emph{specific graph instance}, what is the minimum number of edges that satisfy the hopbound and stretch requirements? We give approximation algorithms for a generalized hopset problem which, when combined with known existential bounds, lead to different approximation guarantees for various regimes depending on hopbound, stretch, and directed vs. undirected inputs.
  We complement our upper bounds with a lower bound that implies Label Cover hardness for directed hopsets and shortcut sets with hopbound at least $3$. 
\end{abstract}

\section{Introduction}

A hopset $H$ with hopbound $\beta$ and stretch $\alpha$ for a given (directed or undirected) graph $G$ is a set of (possibly weighted) edges such that between every pair of vertices $u$ and $v$ in $G$ there is a path with at most $\beta$ hops in $G \cup H$ that approximates the distance between $u$ and $v$ up to multiplicative stretch $\alpha$ (our results hold for a more generalized version of the problem formally defined in Definition \ref{def:gen_hopset}). A related object is shortcut sets, which preserve \emph{reachability} via low-hop paths, rather than distances.
Hopsets were formally introduced 25 years ago by Cohen~\cite{cohen2000}, and they were used to compute approximate single-source shortest paths in undirected graphs in parallel settings.  More recently, they have been shown to be useful for a variety of different problems and settings, and have been studied extensively.  Examples of these applications include low parallel depth single-source shortest paths algorithms \cite{klein1997, cohen2000, miller2015, elkin2019RNC}, distributed shortest-paths computation \cite{elkin2019journal, elkin2017, censor2021}, and dynamic shortest-paths \cite{bernstein2011,henzinger2014, gutenberg2020, chechik2018, LN2022} with implications for fast static flow-based algorithms \cite{madry2010, bernstein2021deterministic, chen2022maximum} and dynamic clustering \cite{cruciani2024dynamic}, faster construction of related objects such as distributed or massively parallel distance sketches \cite{elkin2017, dinitz2019, DM24}, parallel reachability \cite{ullman1990high}, and work-efficient algorithms for reachability and directed shortest paths \cite{cao2020efficient, cao2020improved, cao2023exact,fineman2018nearly,jambulapati2019parallel}.

In addition to their wide range applications, hopsets and shortcut sets are also studied as fundamental objects in their own right. There has been a surge of recent work that, rather than focusing on running time, focuses on finding the best \emph{existential} (extremal) upper or lower bounds for hopsets and shortcut sets \cite{KP22,BW23, hesse2003directed, huang2021lower, kogan2023towards,  williams2024simpler}. 
Namely, the main goal is to find the smallest value $\gamma(n,\beta, \alpha)$ such that \emph{every} graph $G$ on $n$ nodes admits a hopset with hopbound $\beta$ and stretch $\alpha$.

Another line of work has explored the connections between hopsets and many other fundamental objects such as spanners \cite{BP2020, abboud2018}, emulators \cite{elkin2020survey, huang2019}, and distance preservers \cite{KP2022hope, hoppenworth2025, bodwin2023bridge}.

The two main existing lines of work in this direction are on finding efficient algorithms for constructing hopsets in various computational models for specific applications, or on existential bounds and structural properties.
Such results are very useful, since they give us a \emph{worst-case bound} on how many edges we must add to an arbitrary graph if we want a hopset of a certain quality and how fast can this be done for our particular application.  The focus in these results has been on improving these existential bounds and developing fast algorithms for a variety of different settings and parameters; see Section \ref{sec:related_work} for a discussion of known results of this form.

A complementary type of problem with a very different flavor is \emph{optimization}: \emph{given} a graph $G$, hopbound $\beta$, and stretch bound $\alpha$, can we design an algorithm to efficiently find the smallest hopset \emph{for $G$} (rather than in the worst case over all graphs)?  If the minimization problem is NP-hard, then can we \emph{approximate} the smallest hopset for $G$?  Note that the existential and optimization versions of these problems are complementary: good existential bounds guarantee that no matter the graph $G$, there will be a reasonably small hopset; conversely, good optimization results guarantee that we can find an approximately minimum hopset, even if the best hopset for $G$ is significantly smaller than for the worst-case graph.  An existential bound might convince someone to use a hopset for some application---since they know they will never need to add too many edges---but once they commit to using a hopset for that application, they might naturally want to find the best hopset for their particular input. Additionally, in many distributed and parallel settings, adding hopset edges can be seen as a preprocessing step (that can take considerable time), after which (approximate) distance queries can be performed in fewer distributed/parallel rounds (often corresponding to the hopbound). Thus in such cases, we may be willing to spend more preprocessing time in order to add the fewest number of edges.

This natural complementarity between the existential and optimization versions has led to significant study of \emph{both} versions for a number of related objects, most notably graph spanners (subgraphs that approximately preserve distances).  For spanners, while the vast majority of work has been on the existential questions (see, e.g., \cite{ADDJS93} for the fundamental tradeoff between stretch and size), there has also been significant work on the optimization versions (see, e.g., \cite{KP94,KP98,Kor01,EP07,DK11,BBMRY11,DKR16,DNZ20,GKL23,CD16,CDKL20,GLQ21,CDK12,CDR19}).  Similarly, there has also been work on optimization versions of other related objects such as reachability preservers\cite{abboud2024reachability}, 2-hop covers \cite{CHHZ2003}, and diameter reduction\cite{demaine2010minimizing}.

Despite both the recent importance of hopsets and the extensive study of optimization for related objects, the optimization version of hopsets has not yet been considered.  In this paper we initiate this line of research, introducing these optimization variants and proving both upper and lower bounds on their approximability. 

\subsection{Our Results and Techniques}
There are many variants of this problem, divided along three main axes: whether the graph is directed or undirected, what the desired hopbound $\beta$ is, and what the required stretch is (and in particular whether it is arbitrary or whether it is in some particularly nice regime like $1+\epsilon$ for small $\epsilon$ or $2k-1$ for integer $k$).  A summary of our results for these variants can be found in Table~\ref{tab:results}.

\begin{table}[h]
    \centering
    \begin{tabular}{|c|c|c|c|c|}
        \hline
        \textbf{Undirected/Directed} & \textbf{Hopbound $\beta$} & \textbf{Stretch} & \textbf{Approximation} & \textbf{Theorem} \\ \hline

        Directed & $\widetilde{O}(n^{2/5})$  & Individual & $\widetilde{O}(\beta^{1/3} \cdot n^{2/3 + \epsilon'})$ & \ref{thm:small_be_dir_gen} \\ \hline
        
        Directed & $\widetilde{\Omega}(n^{2/5})$ & Individual & $\widetilde{O}(n^{1+\epsilon'}/ \sqrt{\vphantom{\beta^K} \be} \, )$ & \ref{thm:big_be_dir_gen} \\ \hline

        Directed & $\geq 20\log{n}$ & $1+\epsilon$ & $\widetilde{O}(n^{3/4 + \epsilon'} \cdot \epsilon^{-\frac{1}{4}})$ & \ref{thm:dir_eps}  \\ \hline
        
        Directed & 2 & Individual & $O(\ln{n})$ & \ref{thm:2hop-main} \\ \hline

        Undirected & $\widetilde{O}(n^{\frac{1}{2} - \frac{1}{2\eta}})$ & $1+\epsilon$ & $\widetilde{O}(\sqrt{ \be} \cdot n^{\frac{1}{2} + \frac{1}{2\eta} + \epsilon'})$ & \ref{thm:undir_eps} \\ \hline

        Undirected & $\widetilde{O}(k^{-1/2} \cdot n^{\frac{1}{2} - \frac{1}{2k}})$ & $2k-1$ & $\widetilde{O}(\sqrt{k \be} \cdot n^{\frac{1}{2} + \frac{1}{2k} + \epsilon'})$ & \ref{thm:undir_gen_stretch}  \\ \hline

    \end{tabular}
    \caption{Comparison of hopset results, where $\epsilon' > 0$ is an arbitrarily small constant, $\epsilon \in (0,1)$, and $\eta > \be^{1/(\ln{\ln{\be}}-\frac{1}{2}\ln{\ln{\ln{\be}}})}$. For $\be \geq 3$, $\eta \geq 6$. ``Individual'' stretch means arbitrary distance bounds, possibly different for each individual demand pair. All results are for pairwise demands, with edge lengths being nonnegative integer and upper bounded by $\texttt{poly}(n)$. }
    \label{tab:results}
\end{table}

\begin{table}[h]
    \centering
    \begin{tabular}{|c|c|c|c|c|c|}
        \hline
        \textbf{Demands} & \textbf{Edge Lengths} & \textbf{Edge Costs?} & \textbf{Stretch} & \textbf{Approximation} & \textbf{Paper} \\ \hline
        Pairwise & Arbitrary & Y & 2   & $O(\ln n)$   & \cite{DK11}   \\ \hline
        All-Pairs & Arbitrary & N & Arbitrary & $\widetilde{O}(n^{1/2})$   & \cite{BBMRY11}   \\ \hline
        Pairwise & Integer, $\leq \texttt{poly}(n)$ &  Y  & Individual & $\widetilde{O}(n^{4/5+\epsilon})$  & \cite{GKL23}  \\ \hline
        Pairwise & Vector & Y  & Indiv. Vectors & $\widetilde{O}(|D|^{1/2+\epsilon} \cdot \tau^{m-1})$  & \cite{GKL24}  \\ \hline
    \end{tabular}
    \caption{Comparison of related results for directed spanners. The result of~\cite{GKL24} allows for \textit{edge resource consumption vectors} as opposed to edge lengths. $D$ is the set of demands, $\epsilon > 0$ is an arbitrarily small constant, $\tau$ is the largest resource budget, and $m$ is the size of the resource consumption vectors.  ``Individual'' stretch means arbitrary distance bounds, possibly different for each individual demand pair. All variants require nonnegative edge lengths.  }
    \label{tab:spanners}
\end{table}

\subsubsection{Upper Bounds}
While this list may appear a bit complex, it turns out that almost all of the results are generated by trading off several approximation algorithms that we design and analyze with the known existential bounds for the given regime.  In particular, we consider three main approximation algorithms: one based on rounding an LP relaxation, one based on randomly sampling stars, and one based on defining and analyzing an appropriate variant of junction trees \cite{GKL23, GKL24, CDKL20}.  Instead of just analyzing the approximation ratio of each algorithm as a function of $n$ (the traditional point of view), we analyze the approximation ratios as functions of $n$, $OPT$, $\beta$ and the \emph{local neighborhood size}~\cite{DK11,BBMRY11}.  Importantly, none of these algorithms have performance that depends on whether the graph is undirected or directed, or what the allowed stretch is.

To get the results in Table~\ref{tab:results}, we trade these three algorithms off not only with each other based on the described parameters, but also with the known existential bounds.  Unsurprisingly, there are different existential bounds known for each setting.  These different existential bounds are what lead to each of the different rows of Table~\ref{tab:results}.

We also note that there is a simple reduction from hopsets to the multicriteria spanner problem, introduced by~\cite{GKL24}. For multicriteria spanners, instead of edges being associated with just edge lengths---as with traditional spanners---edges in the multicriteria spanner problem each have a resource consumption vector, where each entry corresponds to the consumption of the respective resource on that edge. The simple reduction to multicriteria spanners implies a $\widetilde{O}(|D|^{1/2 + \epsilon} \cdot \beta)$-approximation for hopsets, where $D$ is the set of vertex demand pairs given in the input (see Table~\ref{tab:spanners}).  

\paragraph{Technical Overview.}  Technical components of our main approximation algorithms have appeared before, in particular, in the literature on approximation algorithms for spanners.  The LP rounding algorithm we use is a variant of the one introduced for spanners by~\cite{BBMRY11}, our star sampling is a variant of the arborescence sampling used by~\cite{DK11}, and our junction tree algorithm is a variant of the junction trees used by~\cite{CDKL20} and~\cite{GKL23}.  But using them for hopsets involves overcoming a number of technical difficulties.

First, and most notably, the performance bound on arborescence sampling used by~\cite{DK11,BBMRY11} fundamentally depends on the fact that for traditional graph spanners, the required connectivity implies that $\opt \geq n-1$, and hence polynomial size bounds of the form $n^{\gamma}$ can be interpreted as $n^{\gamma-1} \cdot \opt$, i.e., as $n^{\gamma -1 }$-approximations.  But for hopsets, since we are adding edges rather than removing them, this is not the case: the only bound we have is the trivial bound of $\opt \geq 1$ (or else we are already done).  Hence bounds that are sublinear for spanners (e.g., the $O(n^{1/2})$-approximation of~\cite{BBMRY11}) turn into \emph{superlinear} approximation ratios if translated to hopsets in the obvious way.  These are still nontrivial bounds, in the sense that the trivial bound for hopsets is an $O(n^2)$-approximation  ($\opt \geq 1$ and we can just add all pairwise edges to get a solution of size $O(n^2)$), but superlinear approximation ratios are not particularly satisfying. This is why we need different algorithms for various regimes of $\opt$.
By combining the star sampling and randomized LP rounding we get an approximation bound that is comparatively better for cases where $\beta$ is small and $\opt$ is large. On the other hand, our junction tree algorithm performs better when $\opt$ is smaller.  

Second, while the LP relaxation that we use for our rounding algorithm is a natural variant of the standard spanner LP relaxation (see~\cite{DK11,BBMRY11,DNZ20}), the hopset variant turns out to be significantly more difficult to solve.  Most notably, solving the equivalent LP relaxations for spanners (either the flow version of~\cite{DK11} or the fractional cut version of~\cite{DNZ20}) requires using the ellipsoid method with a separation oracle for a problem known as Restricted Shortest Path, in which we are given two distance metrics (think of one metric as being our original edge lengths and one metric as being the fractional values given by the LP to edges) and are asked to find the shortest path in the second metric subject to a distance constraint in the first metric.  Unfortunately, Restricted Shortest Path is NP-hard, but it turns out that one can use approximation algorithms for Restricted Shortest Path to solve the LP up to any desired accuracy (see~\cite{DK11}).  But for our hopset LP, the equivalent separation oracle has \emph{three} metrics (the original edge lengths, the number of hops, and the LP values), and we need to find the shortest path in the third metric subject to upper bounds in the first two metrics.  So we design a PTAS for this more complex problem in order to solve our LP.

Third, for similar reasons our junction tree argument is more complex. Junction trees were originally introduced for network design problems where the only constraints are on connectivity, not distances; see, e.g., \cite{CEGS11, FKN12}.  By reducing to a ``layered'' graph, Chlamt\'a\v{c} et al.~\cite{CDKL20} showed that junction tree-based schemes are also useful in distance-constrained settings, and this layering technique was pushed significantly further by~\cite{GKL23}.  When trying to use these ideas for hopsets, though, we have the same difficulty as when solving the LP: we have essentially \emph{two} distance constraints for each demand, corresponding to the number of hops and the allowed distance.  To overcome this, we give a ``two-stage'' layering reduction, where we first construct a layered graph that allows us to get rid of the hop requirement without changing the distance requirements.  Then we can use the further layered graph of~\cite{GKL23} to get rid of the distance requirement, transforming it into a pure connectivity problem.  In fact, we can use~\cite{GKL23} as a black-box for this second step.

The black-box nature of our use of~\cite{GKL23} has an interesting corollary: even without caring about junction trees, the layered graph reduction we design will immediately let us use the known results for pairwise spanners \cite{GKL23} (by considering them on the \textit{weighted transitive closure}) to get a non-trivial approximation guarantee for our full problem. However this reduction introduces additional factors in $\beta$, which we improve by white-boxing their approach, which will in turn give us a better trade-off in combination with the other algorithms.

\subsubsection{Lower Bounds}

To complement our upper bounds, we show that (at least for directed graphs) we cannot hope to get approximation ratios that are subpolynomial.  In particular, we give an approximation-preserving reduction from the famous Label Cover problem (or more precisely, its minimization variant Min-Rep~\cite{Kor01}) to the problem of computing the smallest hopset in a directed graph with hopbound at least $3$, for any stretch value.  This gives the following theorem.

\begin{theorem} \label{thm:LB-main}
    Assuming that $NP \not\subseteq DTIME(2^{polylog(n)})$, for any constant $\epsilon > 0$, and for any $\beta \geq 3$, there is no polynomial-time algorithm that can approximate directed {\hopset} (for any stretch value; see Definition~\ref{def:gen_hopset}) or the minimum shortcut set on directed graphs with approximation ratio better than $2^{\log^{1-\epsilon} n}$.
\end{theorem}

Our reduction is quite similar to known hardness reductions used for spanners~\cite{Kor01,EP07,DKR16,CD16}.  The main difficulty is that these previous reductions, since they are for spanners, only have to reason about \emph{subgraphs} of the graph created by the reduction.  That is, given an instance of Min-Rep, they create some graph $G$ and argue that if the Min-REP optimum is large then any subgraph of $G$ that is a spanner must be large.  But for hopsets, not only are the edges of $G$ ``free,'' we also need to argue about hopset edges that \emph{are not} part of $G$.  This requires some changes in the reductions and analysis. The full reduction and details can be found in \iflong Section \fi \ref{app:hardness}.

\subsection{Related Work}\label{sec:related_work}

In earlier applications, sparse $(1+\epsilon)$-approximate hopsets for undirected graphs were used to obtain low parallel depth single-source shortest paths algorithms \cite{klein1997, cohen2000, miller2015, elkin2019RNC}. Similarly, fast hopset constructions for undirected graphs were studied in many other settings such as distributed \cite{elkin2019journal, elkin2017, censor2021, elkin2017, dinitz2019, DM24} and dynamic settings \cite{bernstein2011,henzinger2014, gutenberg2020, chechik2018, LN2022}.

Another line of work focuses on fast computation of hopsets for directed graphs and shortcut sets that preserve pairwise reachability, rather than distances, while reducing the diameter. These have gained attention particularly due to their application in parallel reachability and directed shortest path computation \cite{ullman1990high, cao2020efficient, cao2020improved, cao2023exact, fineman2018nearly,jambulapati2019parallel}.

More recently, hopsets and shortcut sets have been studied from an extremal (or existential) point of view.
In particular, for $(1+\epsilon)$-hopsets, there are almost (up to $1/\epsilon$ factors) matching upper bounds \cite{elkin2019journal, huang2019} and lower bounds \cite{abboud2018} in \emph{undirected graphs}. One trade-off that is widely used in $(1+\epsilon)$-approximate single-source shortest paths applications is that for any graph, there exists a $(1+\epsilon)$-hopset of size $n^{1+o(1)}$ with hopbound $n^{o(1)}$.  On the other hand, \cite{abboud2018} showed that there are instances in which this trade-off is tight. For larger stretch, better size/hopbound trade-offs are known.

In directed graphs there are polynomial gaps between existential lower bounds and upper bounds, both for approximate hopsets and shortcut sets. A widely used folklore sampling approach implies an $\widetilde{O}(n)$ size for exact hopset and shortcut sets, with hopbound $O(\sqrt{n})$ \footnote{The upperbounds for directed hopsets and shortcut sets give a smooth trade-off between size and hopbound. For simplicity, we discuss the important linear-size regime.}.
A breakthrough result of \cite{KP22} went beyond this long-standing bound by constructing shortcut sets of size $\widetilde{O}(n)$ with hopbound $\widetilde{O}(n^{1/3})$. 
Later, \cite{BW23} showed that the same upper bound also holds for directed $(1+\epsilon)$-hopsets (up to log factors in weights and aspect ratio). More on these existential bounds can be found in Section \ref{sec:existential}, where we trade them off with our approximation algorithms. Another line of work focused on obtaining subsequently better existential lower bounds for directed graphs \cite{hesse2003directed, huang2021lower, kogan2023towards,  williams2024simpler}. The current best known lower bound for a linear size shortcut set is $\widetilde{\Omega}(n^{1/4})$ \cite{BH23folklore, williams2024simpler}. 

Another recent result \cite{BH23folklore} showed that for \textit{exact hopsets} both in directed and undirected \textit{weighted} graphs, the folklore sampling (see Section \ref{sec:existential}) is existentially optimal, implying that the separation between approximate hopsets for directed and undirected graphs, does not hold for exact hopsets.

On the approximation algorithms side, most relevant to hopsets are the weighted pairwise spanner algorithms \cite{GKL23} (see Table~\ref{tab:spanners}), 2-hop covers \cite{CHHZ2003}, and multicriteria spanners \cite{GKL24} (Table~\ref{tab:spanners}). In particular, weighted pairwise spanners (like hopsets and unlike $k$-spanners) do not have $n$ as a trivial lower bound on $\opt$, and hence some of the difficulties encountered when designing algorithms are similar. Additionally, when we do not have distance bounds---that is, when stretch is $\infty$ for all demand pairs---weighted pairwise spanners captures the hopset problem.

Approximation algorithms for 2-hop covers~\cite{CHHZ2003}, hub labelings~\cite{CHHZ2003,BGGN17}, and 2-spanners \cite{KP94,DK11} in particular are closely related to our hopbound $2$ regime (see \iflong Section~\ref{sec:2hop}\fi).  In all of these problems, the requirement of covering using 2-paths leads to algorithms that are essentially solving some variant of Set Cover.  Our situation is similar, so we can use similar techniques, but our Set Cover variant is slightly different from these other problems, most notably because we do not ``pay'' for edges that already exist in the graph. 

Multicriteria spanners were recently introduced by \cite{GKL24}; they also rely on a modification of the junction tree framework to give an approximation algorithm. While their junction tree framework works for our setting, we present a different framework tailored to hopsets that yields better approximations for our setting than their framework gives.
Other related problems are optimizations for other variants of spanners, including traditional directed $k$-spanner (unit edge costs, all-pairs demands) \cite{DK11, BBMRY11} (Table~\ref{tab:spanners}), buy-at-bulk spanners~\cite{bulkGKL24}, and transitive closure spanners \cite{BGJRW08, CJMN25}. In addition to transitive closure spanners, \cite{CJMN25} shows hardness of bicriteria approximation for shortcut sets (where they allow the hopbound to be violated).

\section{Preliminaries}
Going forward, we will generally operate on what we call the \textit{weighted transitive closure} of $G$, defined as follows. Let $d_G(s,t)$ denote the distance in $G$ from $s$ to $t$.

\begin{definition}[Weighted Transitive Closure] \label{def:trans_closure}
    The weighted transitive closure of a graph $G$, denoted by $G_M = (V, E_M)$, is the weighted graph obtained by first taking the transitive closure of $G$, then assigning weight $d_G(u,v)$ to each edge $(u,v)$ in the transitive closure. We use $\widetilde{E} = E_M \setminus E$ to denote the set of edges in the weighted transitive closure that are not provided in the input $G$. 
\end{definition}

Formally, we consider the following problem:

\begin{definition}[\hopset]\label{def:gen_hopset}
    Given a directed graph $G = (V, E)$, edge weights (or ``lengths'') $\ell : E \rightarrow \{1,2,3,...,\texttt{poly}(n) \}$, a set of vertex pairs $\mathcal D \subseteq V \times V$, and a distance bound function $Dist: \mathcal{D} \rightarrow \mathbb{N}_{\geq 0}$, find the smallest set of edges $H \subseteq \widetilde{E}$ such that for each vertex pair $(s,t) \in \mathcal{D}$, it is the case that $d_G(s,t) \leq d^{(\beta)}_{H \cup G}(s,t) \leq Dist(s,t)$. In other words, there must be an $s-t$ path $P$ in $H \cup G$ such that $|P| \leq \beta$ and $\sum_{e \in P} \ell(e) \leq Dist(s,t)$.
\end{definition}

If $G$ is directed, we say that the problem of interest is \textit{directed} {\hopset}; otherwise it is \textit{undirected} {\hopset}. Additionally, if for all $(s,t) \in \mathcal{D}$ we have that $Dist(s,t) = k \cdot d_G(s,t)$ for some $k$, then this is the \textit{stretch-$k$} {\hopset} problem. Note that {\hopset} is a generalized version of the traditional hopset problem: all of our results hold for \textit{any} set of vertex demand pairs, and many of our results hold for arbitrary distance bound functions.

We will use ${\opt}$ to refer to the cost of the optimal solution to the {\hopset} problem instance. That is, {\opt} will refer to the number of edges in the optimal hopset. A path is an \textit{$i$-hop path} if there are exactly $i$ edges on the path. We also say that a demand $(s,t) \in \mathcal{D}$ is \textit{settled} or \textit{satisfied} by a graph $G$ (or an edge set $E$) if there is a path from $s$ to $t$ in $G$ (in $E$) with at most $\be$ hops and distance at most $Dist(s,t)$. Otherwise, demand $(s,t)$ is considered \textit{unsettled} or \textit{unsatisfied}.

Note that the weighted transitive closure of a graph can be found in polynomial-time. When working in the weighted transitive closure, we will generally assign costs to the edges in $G_M$: edges in $E$ will have cost $0$, while edges in $\widetilde{E}$ will have cost $1$. It is easy to see that {\hopset} on $G$ is equivalent to finding a min-cost subgraph of $G_M$ that settles all demand pairs. We will use $c(F)$ to denote the \textit{cost} of an edge set $F$; namely, $c(F) = |F \cap \widetilde{E}|$ is the number of edges in $F$ not provided in the input. 
Finally, we will assume each edge $(u,v) \in E$ in the input is a shortest path from $u$ to $v$ in $G$.

\section{LP Relaxation} \label{sec:lp_relaxation}

In this section, we state and solve a cut-covering linear program (LP) for {\hopset}. A few of our approximation algorithms will randomly round the solution to this LP as a subroutine. 

Let $\mathcal{P}_{s,t}$ be the set of paths from $s$ to $t$ that have at most $\beta$ hops and path length at most $Dist(s,t)$. We will refer to these paths as ``allowed'' or ``valid'' paths. The natural flow LP for our problem requires building enough capacity to send one unit of (non-interfering) flow along valid paths for each demand; this is the basic LP used in essentially all network design problems, and was introduced for spanners by~\cite{DK11}.  An equivalent LP, which is what we will use for {\hopset}, is obtained through the duality between flows and fractional cuts (of valid paths).  This version of the LP for spanners was studied by~\cite{DNZ20}, and strengthens the anti-spanner LP of~\cite{BBMRY11}. 

In more detail, for a graph $G = (V,E)$, we say that an \quotes{$s-t$ cut against valid paths} is a set of edges $F$ such that in the graph $G \setminus F$, there are no valid paths from $s$ to $t$. In the cut covering LP we will use, the constraints will ensure that any feasible solution must \quotes{cover} every cut against valid paths; that is, any feasible solution must buy an edge in each of these cuts. This leads to our LP relaxation, which we call~\ref{lp:hopset}. The LP has a variable $x_e$ for every edge $e \in E_M$. Note that because edges in $E$---edges from the input graph---do not contribute to the cost of the solution, we can assume without loss of generality that $x_e = 1$ for all $e \in E$.  
Let $\mathcal{Z}_{s,t} = \{ \bm{\mathrm{z}} \in [0,1]^{|E|} : \forall P \in \mathcal{P}_{s,t} \; , \; \sum_{e \in P} z_e \geq 1 \}$ be the set of vectors representing all fractional cuts against valid $s-t$ paths. For each demand, our LP requires any feasible integral solution to buy at least one edge from every cut against valid paths.

\begin{equation} \label{lp:hopset} \tag{LP(Hopset)}
\fbox{$
\begin{array}{lll}
\min & \displaystyle \sum_{e \in \widetilde{E} } x_e \\
\mathrm{s.t.} & \displaystyle \sum_{e \in E_M} z_e x_e \geq 1 \qquad \qquad & \forall (s,t) \in \mathcal{D}, \, \forall \bm{\mathrm{z}} \in \mathcal{Z}_{s,t} \\
& \displaystyle x_e \geq 0 & \forall e \in E_M \\
\end{array}
$
}
\end{equation}

The two main differences between the $k$-spanner cut-covering LP and ours are 1) for spanners, valid paths are those that satisfy the stretch constraint, whereas for hopsets, we want paths that satisfy both distance and hop constraints, and 2) rather than a subgraph of $G$, we are interested in a subgraph of the \textit{weighted transitive closure} of $G = (V,E)$ (see Definition~\ref{def:trans_closure}). 

As written, \ref{lp:hopset} has an infinite number of constraints. Consider the polytope $\mathcal{Z}_{s,t}$ for some demand $(s,t)$. Due to convexity, we only need to keep the constraints that correspond to the vectors $\bm{\mathrm{z}}$ that form the vertices of of $\mathcal{Z}_{s,t}$. Since there are only an exponential number of these constraints, we can assume the LP has exponential size. \iflong We now quickly prove that the LP is a valid LP relaxation of {\hopset}. \fi

\iflong
\begin{lemma}
    Let $H$ be a feasible solution to {\hopset}. For every edge $e \in E_M$, let $x_e = 1$ if $e \in H$ or $e \in E$. Otherwise let $x_e = 0$. Then, \bm{\mathrm{x}} is a feasible integral solution to~\ref{lp:hopset}.
\end{lemma}
\begin{proof}
    Clearly, $x_e \geq 0$ for all $e \in E_M$, so it is left to show that the cut covering constraints are satisfied. For some demand $(s,t) \in \mathcal{D}$, consider some edge cut against valid $s-t$ paths, $C$. There is a vector $\bm{\mathrm{z}} \in \mathcal{Z}_{s,t}$ that serves as the indicator vector for cut $C$---specifically, there is a vector $\bm{\mathrm{z}} \in \mathcal{Z}_{s,t}$ such that for every edge $e \in C$, we have $z_e = 1$, and $z_e = 0$ otherwise. 
    
    Since $H$ is a feasible solution to {\hopset}, there must be some edge $e \in C \cap H$ (otherwise, there is no valid path from $s$ to $t$ in $H \cup E$, making $H$ infeasible). This also means $x_e = 1$ and  $z_e = 1$. Hence we have that for $\bm{\mathrm{z}}$, the constraint $\sum_{e \in E_M} z_e x_e \geq 1$ is satisfied.
\end{proof}

\begin{lemma}
    Let $\bm{\mathrm{x}}$ be a feasible integral solution to~\ref{lp:hopset}, and let $H = \{ e \; | \; x_e = 1, e \notin E\}$. Then $H$ is a feasible solution to {\hopset}.
\end{lemma}
\begin{proof}
    Suppose for contradiction there is some demand, $(s,t) \in \mathcal{D}$, that is not satisfied by $H$.  Then there is some cut $C$ against valid $s-t$ paths such that $H \cap C = \emptyset$ (and $x_e = 0$ for all $e \in C$). Since $C$ is a cut against valid $s-t$ paths, there is an indicator vector $\bm{\mathrm{z}} \in \mathcal{Z}_{s,t}$ such that for every edge $e \in C$, we have $z_e = 1$ ($z_e = 0$ otherwise). Therefore we have that $\sum_{e \in E_M} z_e x_e = 0 < 1$, violating the constraint in \ref{lp:hopset}, and thus the assumption that $\bm{\mathrm{x}}$ is feasible. Hence $H$ must be feasible.
\end{proof}
\fi

\subsection{Solving the LP}

Some of our algorithms for {\hopset} involve some flavor of random LP rounding, so in this section we will (approximately) solve~\ref{lp:hopset}, our LP of interest. The LP has an exponential number of constraints, so we must solve it using the ellipsoid algorithm with a separation oracle. To do so, we start by defining another LP, \ref{lp:oracle1}, that captures the problem of finding a violated constraint of~\ref{lp:hopset}. To solve~\ref{lp:oracle1}, we find yet another (approximate) separation oracle. This second oracle boils down to solving a hopbounded version of the Restricted Shortest Path problem. We show that this problem admits an FPTAS, and that this ultimately translates to a $(1+\epsilon)$-approximation for~\ref{lp:hopset}. More formally, we will prove the following:

\begin{theorem} \label{thm:solve_LP}
    There is a $(1+\epsilon)$-approximation to the optimal solution of~\ref{lp:hopset} for any arbitrarily small constant $\epsilon > 0$.
\end{theorem}

\subsubsection{Oracle 1}
As noted, \ref{lp:hopset} has exponential size, so we find a separation oracle in order to run ellipsoid. To do so, we must determine if there is a fractional cut $\bm{\mathrm{z}}$ that is not covered by the LP solution $\bm{\mathrm{x}}$---that is, given $\bm{\mathrm{x}}$ (satisfying the non-negativity constraints), we want to find $\bm{\mathrm{z}} \in \mathcal{Z}_{s,t}$, for some demand $(s,t) \in \mathcal{D}$, such that $\sum_{e \in E_M} z_e x_e < 1$. We can find such a violated cut-covering constraint by solving the following LP for each demand $(s,t) \in \mathcal{D}$.

\begin{equation} \label{lp:oracle1} \tag{LP(Oracle 1)}
\fbox{$
\begin{array}{lll}
\min & \displaystyle \sum_{e \in E_M} z_e x_e \\
\mathrm{s.t.} & \displaystyle \sum_{e \in P} z_e \geq 1 \qquad \qquad & \forall P \in \mathcal{P}_{s,t} \\
& \displaystyle z_e \geq 0 & \forall e \in E_M
\end{array}
$
}
\end{equation}

\subsubsection{Oracle 2: Hopbounded Restricted Shortest Path Problem}
\ref{lp:oracle1} is also exponential in the number of constraints, so we use the ellipsoid algorithm again, with yet another separation oracle. Given an LP solution $\bm{\mathrm{z}}$, we must determine whether there is some valid $s-t$ path that is not fractionally cut (or, covered) by $\bm{\mathrm{z}}$. Specifically, we need to find a violated constraint of the form $\sum_{e \in P} z_e < 1$ for some path $P \in \mathcal{P}_{s,t}$. Observe that we now have three separate metrics: the $\bm{\mathrm{z}}$ metric (given by a solution to \ref{lp:oracle1}), our original distance metric from the input, and a hop metric (where each edge has hop ``length'' 1). We only care about valid $s-t$ paths; namely, paths with length at most $D(s,t)$ in our original distance metric and with length at most $\beta$ in our hop metric. Thus, our goal is to find the shortest path in the $\bm{\mathrm{z}}$ metric that respects these upper bounds in the distance and hop metrics. 

When there are only two metrics---that is, when the goal is to find the shortest path in one metric subject to an upper bound in the other---this is the Restricted Shortest Path problem, defined as follows.

\begin{definition} [Restricted Shortest Path Problem]
    Let $G = (V,E)$ be a graph such that each edge $e \in E$ is associated with a cost $z_e$ and a delay $\ell_e$. Let $T$ be a positive integer, and $s,t \in V$ be the source and target nodes, respectively. The Restricted Shortest Path problem is to find a path $P$ from $s$ to $t$ such that the delay along this path is at most $T$, and the cost of $P$ is minimal.
\end{definition}

The Restricted Shortest Path problem is well studied \cite{Has92, War87, Phi93, XZTT08, HK18}. The problem admits an FPTAS, meaning there exists a polynomial-time (1+$\epsilon$)-approximation for the problem \cite{LR01}. Since our problem of interest has a third metric (the hop metric), we refer to it as the \textit{Hopbounded} Restricted Shortest Path problem. We \iflong will \else \fi modify the Restricted Shortest Path FPTAS algorithm of \cite{LR01} to give an FPTAS for the Hopbounded Restricted Shortest Path problem, thus giving a $(1+\epsilon)$-approximate separation oracle for~\ref{lp:oracle1}. We formally define the Hopbounded Restricted Shortest Path problem with respect to~\ref{lp:oracle1}:

\begin{definition}[Hopbounded Restricted Shortest Path problem]
    Let $G_M = (V,E_M)$ be a graph such that each edge $e \in E$ is associated with a cost $z_e$ and a delay, or length, $\ell_e$. Let $T = D(s,t)$, and $s,t \in V$ be the source and target nodes, respectively. The Hopbounded Restricted Shortest Path problem is to find a path $P$ from $s$ to $t$ with at most $\beta$ hops, such that the length along this path is at most $T$, and the cost of $P$ is minimal.
\end{definition}

\subsubsection{Hopbounded Restricted Shortest Paths Algorithm}
At a high level, the Restricted Shortest Path algorithm works as follows. The algorithm runs multiple binary searches to find good upper and lower bounds on the cost of the optimal solution, then uses these bounds to scale and discretize (or ``bucket'') the costs of the edges. They then give a pseudo-polynomial time dynamic programming algorithm on the problem with bucketed edge costs, which they show is a $(1+\epsilon)$-approximation for the original problem. For runtime, the pseudo-polynomial time DP algorithm is polynomial in $U/L$ (and in $1/\epsilon$), where $U$ and $L$ are the upper and lower bounds, respectively, on the optimal solution. The bounds they find are such that $U/L = O(n)$, and so the algorithm is an FPTAS. 

To get this algorithm to work for the Hopbounded Restricted Shortest Path problem, we simply modify the DP algorithm to take our hop metric into account, adding a factor $\beta$ to the runtime (see Algorithm~\ref{alg:HRSP}). 
\iflong
Like the Restricted Shortest Path DP algorithm, our DP algorithm takes in as input valid upper and lower bounds, $U$ and $L$, respectively. We also scale the edge costs in the same way, and refer to these scaled costs as \quotes{pseudo-costs.} Let $D(v, i, j)$ indicate the minimum length (in the original distance metric) of a path from vertices $s$ to $v$ with pseudo-cost at most $i$ and at most $j$ hops. Our algorithm differs from the Restricted Shortest Path algorithm in that we add an additional dimension for hops to the dynamic programming algorithm.

\begin{algorithm}[ht]
\DontPrintSemicolon

\textbf{Input:} Graph $G = (V, E)$, demand pair $s,t \in V$, \\
cost, distance, and hop metrics $\{z_e, \ell_e, h_e \}_{e \in E_M}$, \\
parameters $T, \beta, L, U,$ and $\epsilon $ \;~\\

Let $S \gets \frac{L \epsilon}{n+1}, \; \Tilde{U} \gets \floor*{U/S} + n + 1$ \\ \; 

\tcp{scale edge costs} 
\ForEach{edge $e \in E$}{
    Let $\Tilde{z}_e \gets \floor{z_e / S} + 1$} \;

\tcp{set base cases} 
\ForEach{index $i  = 0, 1, \dots, T$}{
    $D(s, i, 0) \gets 0$} 
\ForEach{index $j  = 0, 1, \dots, \beta$}{
    $D(s, 0, j) \gets 0$} 
    
\ForEach{vertex $v \in V$ such that $v \neq s$}{
    \ForEach{index $i = 0, 1, \dots, T$}{
        $D(v, i, 0) \gets \infty$}
    \ForEach{index $j = 0, 1, \dots, \beta$}{
        $D(v, 0, j) \gets \infty$} 
} \;

\tcp{build up DP table}
\ForEach{index $i = 1, 2, \dots, \Tilde{U}$}{
    \ForEach{index $j = 1, 2, \dots, \beta $}{
        \ForEach{vertex $v \in V$}{
            $D(v, i, j) \gets \min\{ \: D(v, i-1, j), \; D(v, i, j-1) \: \}$ \;
            \ForEach{edge $e \in \{ (u,v) \; | \; \Tilde{z}_{(u,v)} \leq i \}$ }{
                $D(v, i, j) \gets \min\{ \: D(v, i, j), \; \ell_e + D(v, i-\Tilde{z}_{e}, j-1) \: \}$ \;
            }
        \If{$D(t, i, j) \leq T$}{
            \textbf{Return} the corresponding path \;
        }
        }
    }
}
\textbf{Return} FAIL \;

\caption{\label{alg:HRSP} Hopbounded Restricted Shortest Path Algorithm } 
\end{algorithm}
Using arguments from \cite{LR01}, we will show that Algorithm~\ref{alg:HRSP} is a $(1+\epsilon)$-approximation of the Hopbounded Restricted Shortest Path problem.
\else
\fi
\iflong The following is directly from \cite{LR01} and also holds for our slightly modified algorithm. Let $z(P)$ denote the cost (in the $\bm{\mathrm{z}}$ metric) of a path $P$ in $G_M$, and let $z^*$ denote the cost of the optimal path for the original Hopbounded Restricted Shortest Path instance. 

\begin{lemma}[Lemma 3 of \cite{LR01}]
\label{lem:eps_approx}
    If $U \geq z^*$, then Algorithm~\ref{alg:HRSP} returns a feasible path, $P'$, that satisfies $z(P') \leq z^* + L\epsilon$
\end{lemma}

\cite{LR01} uses this Lemma, and others, to show that there is a $(1+\epsilon)$-approximation for the Restricted Shortest Path problem (Theorems 3 and 4 of \cite{LR01}). Their algorithm achieves a runtime of $O(m n/ \epsilon + m n \log \log \frac{U}{L})$. While their full argument would also give a similar runtime for our setting (with an additional factor $\beta$), we only include the arguments necessary to achieve a polynomial runtime. As a result, we get a worse runtime than what \cite{LR01} would imply. 

\begin{lemma} \label{lem:AS}
    Given valid lower and upper bounds $0 \leq L \leq z^* \leq U$, Algorithm~\ref{alg:HRSP} is a $(1+\epsilon)$-approximation for the Hopbounded Restricted Shortest Path problem that runs in time $O(\beta m n U / L \epsilon)$. 
\end{lemma}
\begin{proof}
The approximation ratio is directly implied by Lemma~\ref{lem:eps_approx}. 
    
Each edge in the algorithm is examined a constant number of times, so the overall time complexity is 
\begin{align*}
    O(\beta m \Tilde{U}) &= O\left( \beta m \frac{n}{\epsilon} \frac{U}{L} + nm \right) \\
    &= O\left( \beta m \frac{n}{\epsilon} \frac{U}{L}  \right) \tag{$U \geq L$ and $\epsilon \leq 1$}
\end{align*}
\end{proof}

Now we find upper and lower bounds $U$ and $L$ such that $U/L \leq n$, just as in \cite{LR01}, to give an overall runtime of $O(\beta m n^2 / \epsilon)$ for our problem. Let $\ell_1 < \ell_2 < \dots, \ell_p$ be the distinct lengths of all edges in $E_M$ (note that $p \leq m$). Let $E_i$ be the set of edges in $E_M$ with length at most $\ell_i$, and let $G_i = (V, E_i)$. We also set $E_0 = \emptyset$.
There must be some index $j \in [1,p]$ such that there exists a $T$-length-bounded, $\beta$-hopbounded path in $G_j$, but not in $G_{j-1}$. As observed in Claim 1 of \cite{LR01}, this means that $\ell_j \leq z^* \leq n \ell_j$, giving bounds $L = \ell_j$ and $U = n \ell_j$, with $U/L = n$. To find $\ell_j$, we can binary search in the range $\ell_1, \dots, \ell_p$, running a shortest hopbounded path algorithm at each step of the binary search to check if there exists a $T$-length-bounded, $\beta$-hopbounded $s-t$ path on the subgraph $G_i$ corresponding to that step. The binary search requires $O(\log m)$ steps, and the shortest hopbounded path algorithm takes $O(m \log n)$ time, so the runtime for finding these bounds is dominated by the runtime of the the DP algorithm. 

In summary, the Hopbounded Restricted Shortest Path algorithm is as follows: Run binary search on $\ell_1, \dots, \ell_p$ to find $\ell_j$ and get bounds $U$ and $L$. Then run the dynamic programming algorithm, Algorithm~\ref{alg:HRSP}, using bounds $U$ and $L$.
\else \fi

\begin{lemma} \label{lem:FPTAS_HRSP}
    The Hopbounded Restricted Shortest Path problem admits an FPTAS.
\end{lemma}

\subsubsection{Proof of Theorem~\ref{thm:solve_LP}}
With an FPTAS for the Hopounded Restricted Shortest Path problem (by Lemma~\ref{lem:FPTAS_HRSP}), we have an approximate separation oracle for~\ref{lp:oracle1}. Using the ellipsoid algorithm with this oracle, we find a solution $\bm{\mathrm{z}}$ for~\ref{lp:oracle1}. While $\bm{\mathrm{z}}$ may not be feasible, it only violates each constraint by a factor of at most $(1+\epsilon)$. That is, $\bm{\mathrm{z}}$ satisfies $(1+\epsilon) \sum_{e \in P}  z_e \geq 1$ for all $P \in \mathcal{P}_{s,t}$. Thus if we scale $\bm{\mathrm{z}}$ by $(1+\epsilon)$, we get a feasible solution. Let $\bm{\mathrm{z'}}$ be this scaled solution, where $z_e' = (1+\epsilon) z_e$ for all $e \in E_M$. We then also have that $\bm{\mathrm{z'}}$ is a $(1+\epsilon)$-approximation for~\ref{lp:oracle1}. This is implied by the fact that for any feasible solution $\bm{\mathrm{z''}}$ of~\ref{lp:oracle1}, the value of the objective on $\bm{\mathrm{z}}$ is at most the value of the objective on $\bm{\mathrm{z''}}$ (the entire feasible region satisfies the $(1+\epsilon)$ \quotes{approximate} constraints, and therefore the feasible region is in the search space of the ellipsoid algorithm). That is, $\sum_{e \in E_M}  z_e x_e \leq \sum_{e \in E_M}  z''_e x_e$ for all feasible solutions $\bm{\mathrm{z''}}$.  Thus we have a $(1+\epsilon)$-approximation for~\ref{lp:oracle1}.

The purpose of solving~\ref{lp:oracle1} is to give a separation oracle for our starting LP, \ref{lp:hopset}. We therefore must show that the $(1+\epsilon)$-approximation we have for~\ref{lp:oracle1} translates to a $(1+\epsilon)$-approximation for the hopset LP. The approach is similar to how we approximate~\ref{lp:oracle1}. Running ellipsoid with the~\ref{lp:oracle1} approximation as a separation oracle, we get a solution $\bm{\mathrm{x}}$ that approximately (within a factor $(1+\epsilon)$) satisfies the constraints. We can scale $\bm{\mathrm{x}}$ by a factor $(1+\epsilon)$ to get a feasible solution $\bm{\mathrm{x'}}$. Additionally, we have $\sum_{e \in E_M} x_e \leq \sum_{e \in E_M} x''_e$ for all feasible solutions $\bm{\mathrm{x''}}$. We therefore have a $(1+\epsilon)$-approximation for the hopset LP relaxation.


\iflong
\section{Approximation for Hopbound 2} \label{sec:2hop}
When the hopbound is $2$, we can get improved bounds by using the hopset version of a spanner approximation algorithm.  Interestingly, while we will require the hopbound to be $2$, we can handle arbitrary distance bound functions (including exact hopsets, $(1+\epsilon)$-stretch hopsets, and larger stretch hopsets), in contrast with the spanner version which requires all edges to have unit length and requires the \emph{stretch} to be at most $2$.

We use a version of the LP rounding algorithm for stretch $2$ from~\cite{DK11}.  In particular, we first solve our LP relaxation~\eqref{lp:hopset} to get an optimal solution $\bm{\mathrm{x}}$, and then round as follows.  First, every vertex $v \in V$ chooses a threshold $T_v \in [0,1]$ uniformly at random.  We then return as a hopset the set of edges $E' = \{(u,v) : \min(T_u, T_v) \leq (c \ln n) \cdot x_{(u,v)} \}$, where $c$ is an appropriately chosen constant. 

\begin{lemma} \label{lem:2hop-cost}
    $\E[c(E')] \leq O(\ln n) \cdot \opt$
\end{lemma}
\begin{proof}
    Fix some $e = (u,v) \in \widetilde{E}$. 
    Clearly $\Pr[e \in E'] \leq (2c\ln n) x_e$, and hence $\E[c(E')] \leq \sum_{e \in \widetilde{E}} (2c\ln n) x_e \leq (2c \ln n) \cdot \opt$.
\end{proof}

\begin{lemma} \label{lem:2hop-correct}
$E'$ is a valid hopset with high probability.
\end{lemma}
\begin{proof}
    Fix some $(u,v) \in V \times V$.  If $x_{(u,v)} \geq 1/(c \ln n)$ then our algorithm will directly include the $(u,v)$ edge.  So without loss of generality, we may assume that $x_{(u,v)} \leq 1/2$.  As discussed, \eqref{lp:hopset} is equivalent to the flow LP, and hence $\bm{\mathrm{x}}$ supports flows $f : \mathcal P_{u,v} \rightarrow [0,1]$ such that $\sum_{P \in \mathcal P_{u,v}} f_P \geq 1$.  Since $\beta = 2$ and $x_{(u,v)} \leq 1/2$, this means that $\sum_{P \in \mathcal P_{u,v} \setminus \{(u,v)\}} f_P \geq 1/2$, i.e., at least $1/2$ flow is sent on paths of length $2$ in $\mathcal P_{u,v}$.  Let $W \subseteq V \setminus \{u,v\}$ be the set of nodes that are the middle node of such a path, and for each $w \in W$ let $P_w$ denote the path $u-w-v \in \mathcal P_{u,v}$.

    Fix $w \in W$.  Note that if $T_w \leq (c \ln n) f_{P_w}$ then our algorithm will include the path $P_w$, since the $\bm{\mathrm{x}}$ variables support the flow $f$.  Hence the probability that we fail to include every path $P \in \mathcal P_{u,v}$ is at most
    \begin{align*}
        \prod_{w \in W} (1-(c\ln n) f_{P_w}) \leq \exp\left(-\sum_{w \in W} (c \ln n) f_{P_w}\right) \leq \exp(-(c \ln n) (1/2)) = n^{-c/2}.
    \end{align*}
    
    We can now take a union bound over all $(u,v) \in V \times V$ to get that the probability that $E'$ is not a valid hopset is at most $n^{2-(c/2)}$.  Hence by changing $c$, we can make the failure probability $1/n^{c'}$ for arbitrarily large constant $c'$.  
\end{proof}

\begin{theorem} \label{thm:2hop-main}
    There is an $O(\ln n)$-approximation algorithm for {\hopset} when $\be = 2$.
\end{theorem}
\begin{proof}
    This is directly implied by Lemmas~\ref{lem:2hop-cost} and \ref{lem:2hop-correct}.
\end{proof}
\else
\fi

\section{Approximation Algorithms for General Hopbounds}
Continuing the connection to spanners, there is a reduction from {\hopset} to the directed pairwise weighted spanner problem (where ``weighted'' refers to edge costs). The most general version of this problem, studied by~\cite{GKL23}, allows for any demand set, any positive rational edge costs, integer edge weights in $\texttt{poly}(n)$, and arbitrary distance bound functions. The reduction starts with the transitive closure $G_M$, and builds a layered graph with $\beta + 1$ copies of each node in $G_M$, and $\be$ copies of each edge (see Section~\ref{sec:layered_reduction} for a more detailed description of the reduction). Since~\cite{GKL23} achieved an $\widetilde{O}(n^{4/5 + \epsilon})$-approximation for directed pairwise weighted spanners, this reduction immediately gives an $\widetilde{O}((\be n)^{4/5 + \epsilon})$-approximation for {\hopset}. 

In this section, we improve upon this result and get an $\widetilde{O}(n^{4/5 + \epsilon})$-approximation for {\hopset}, removing the dependence on $\be$. We will give approximation algorithms in terms of $n, \be, {\opt}$, and ``local neighborhood size.'' All of our algorithms are based on spanner algorithms, and we must modify them (and provide different analyses) to accommodate the hop constraint. We will then trade these algorithms off with known existential hopset results to get approximations (in terms of $n$ and $\be$) in many regimes, including the general setting (directed graphs, arbitrary stretch), and in more restricted settings, such as small and large hopbound, and specific stretch regimes. The approximation  we get in any given regime will be a function of how good the existential bounds are for that regime, so better existential results will result in better approximations.

Our first algorithm, the junction tree algorithm of Section~\ref{sec:junction_tree}, will perform best when ${\be}$ and ${\opt}$ are relatively small. Our second and third algorithms, the star sampling and randomized LP rounding algorithms of Section~\ref{sec:star_sampling_rounding}, will together give better approximations as ${\opt}$ gets larger. 

\subsection{Junction Tree Algorithm} \label{sec:junction_tree}

In this section, we prove the following theorem for directed {\hopset}.

\begin{theorem} \label{thm:junction_tree}
    There is a polynomial-time $\widetilde{O} (\beta n^\epsilon \cdot {\opt})$-approximation for directed {\hopset}.
\end{theorem}

To prove the theorem, we give an algorithm similar to a subroutine of the directed pairwise weighted spanner algorithm of~\cite{GKL23}, where ``weighted'' refers to edge costs. Just as for hopsets, the directed pairwise weighted spanner problem does not have $n-1$ as a lower bound on the cost of the optimal solution. This allows their techniques to be useful in our setting. 

In the pairwise weighted spanner subroutine of~\cite{GKL23}, they define a variant of the junction tree (the ``distance-preserving junction tree''). Junction trees are rooted trees that satisfy demands, and \textit{good} junction trees satisfy many demands at low cost; that is, they have low ``density.'' Junction trees have been used in several spanner approximation algorithms (e.g.~\cite{GKL23, CDKL20, GKL24}). In~\cite{GKL23}, they give an algorithm that iteratively buys their version of low density junction trees until all demands are satisfied. Our algorithm will follow the same structure. The main technical work in this section is in showing that low-density \textit{hopbounded} junction trees exist in our setting, and that we can use the subroutine of~\cite{GKL23} to find these hopbounded junction trees, even though their subroutine does not have any hop guarantees.

We note that the junction tree framework developed by~\cite{GKL24} for multicriteria spanners also works for our setting. Their framework implies an $\widetilde{O}(|\mathcal{D}|^\epsilon \cdot \beta)$-approximation for finding the hopbounded minimum-density junction tree, and thus a $\widetilde{O} (\beta^2 \cdot |\mathcal{D}|^\epsilon \cdot {\opt})$-approximation for directed {\hopset}.
By tailoring our framework to hopsets, we achieve an $O(n^\epsilon)$-approximation for finding the hopbounded min-density junction tree (note the lack of dependence on $\beta$), implying our main result of this section (Theorem~\ref{thm:junction_tree}), a $\widetilde{O} (\beta \cdot n^\epsilon \cdot {\opt})$-approximation overall. By using our hopbounded junction tree framework, we lose a factor $\beta$ in our overall approximation compared to what is implied by~\cite{GKL24}.

We first define a hopbounded variant of the junction tree, which we call an $(i,j)$-distance-preserving hopbounded junction tree. We parameterize by $i,j$, where $i+j \leq \beta$, to ensure that both ``sides'' of the rooted tree---the in-arborescence and the out-arborescence that make up the tree---are hopbounded by $i$ and $j$, respectively, so that all paths in the tree have at most $\beta$ hops. 

\begin{definition}[$(i,j)$-Distance-Preserving Hopbounded Junction Tree] 
     An \textit{$(i,j)$-distance-preserving hopbounded junction tree}, where $i + j \leq \beta$, is a subgraph of $G_M$ that is a union of an in-arborescence and an out-arborescence, both rooted at some vertex $r \in V$, with the following properties: 1) every leaf of the in-arborescence has a path of at most $i$ hops to $r$, 2) for every leaf $w$ in the out-arborescence, there is a path of at most $j$ hops from $r$ to $w$, and 3) for some node $s$ in the in-arborescence and some node $t$ in the out-arborescence, there is an  $s-t$ path through $r$ with distance at most $Dist(s,t)$. The \textit{density} of an $(i,j)$-distance-preserving hopbounded junction tree $T$ is the ratio of the cost of $T$ to the number demands settled by $T$. 
\end{definition}

Going forward, we will refer to the $(i,j)$-distance-preserving hopbounded junction tree as simply an ``$(i,j)$-junction tree.'' Our algorithm will find and remove a low-density hopbounded junction tree from $G_M$, add its edges to the current solution, and repeat, until all demand pairs are settled. 
We will give an $O(n^\epsilon)$-approximation for finding these low-density junction trees. The algorithm will return a subgraph with total cost of $\widetilde{O} (\beta^2 n^\epsilon \cdot {{\opt}}^2 )$.

\subsubsection{Existence of Low-Density Junction Trees}
Let $D$ be the set of \textit{unsettled} vertex pairs at some iteration of the algorithm.
We first argue that a hopbounded junction tree with density $O (\beta^2 \cdot {\opt}^2 / \mbox{ } |D| )$ always exists (at any iteration), where ${\opt}$ is the cost of the optimal solution to the problem instance. 

\begin{lemma}
\label{lem:existence}
    For any set of unsettled demands $D$, there exists an $(i,j)$-junction tree with density $O(\beta \cdot {{\opt}}^2 / |D|) )$.
\end{lemma}
\begin{proof}
    Let $H$ be an optimal solution subgraph to the graph $G_M$, and let $S$ be the cheapest subgraph of $G_M$ that settles all demands in $D$.
    We have that $c(S) \leq c(H) = {\opt}$.
    We now look at the set of paths in $S$ that settle the demands in $D$. Each of these $|D|$ paths must use some edge in $S \cap \widetilde{E}$; otherwise, the demand is already settled and cannot be in $D$. Due to averaging, there must be some edge $e \in S \cap \widetilde{E}$ that belongs to at least $|D| / |S \cap \widetilde{E}| \geq |D| / {\opt}$ of these paths. 
    
    Let $D_e \subseteq D$ be the set of demands settled by paths through $e = (u,v)$ in $S$. Let $P_{e}$ be these demand-settling paths (that is, $P_{e}$ is the set of paths that settle the demands in $D_e$ through $e$). 
    We now place the demands in $D_e$ into $O(\beta)$ classes as follows. Let $x,y$ be nonnegative integers such that $x+y = \beta$. We say that a demand $(s,t) \in D_e$ is in class $(x,y)$ if its corresponding path in $P_{e}$ has at most $x$ hops from $s$ to $u$ and at most $y$ hops from $u$ to $t$. Note that every demand in $D_e$ belongs to at least one class. Recall that $e$ belongs to at least $|D| \mbox{ } / \mbox{ } {\opt}$ demand-settling paths in $S$.
    Again due to averaging, there must be some class containing $\Omega\left( |D| \mbox{ } / \mbox{ } (\beta \cdot {\opt}) \right)$ demands. Let $(i,j)$ denote such a class, let $D_e^{(i,j)} \subseteq D_e$ denote the set of demands in this class, and let $P_e^{(i,j)}$ be their corresponding demand-settling paths. 

    Now consider the tree created by adding all paths in $P_{e}^{(i,j)}$, rooted at vertex $u$. This tree is an $(i, j)$-junction tree, as it is the union of an in-arborescence where each leaf has a path of at most $i$ hops to $u$, and an out-arborescence where $u$ has a path with at most $j$ hops to each leaf, where $i+j \leq \beta$. This tree has cost at most ${\opt}$ and settles at least $|D_e^{(i,j)}| = \Omega \left( |D| \mbox{ } / \mbox{ } (\beta \cdot {\opt}) \right)$ demands, and thus has density $O( \beta \cdot {{\opt}}^2 / \mbox{ } |D| )$.
\end{proof}

\subsubsection{Layered Graph Reduction} \label{sec:layered_reduction}
We want to show that we can \textit{find} a junction tree with low enough density at each iteration of the algorithm. To do so, we will use the junction-tree finding subroutine provided in~\cite{GKL23}. Their subroutine, however, finds non-hop-constrained junction trees. We therefore transform our input graph in order to use their subroutine to find $(i,j)$-junction trees. We build the following $\beta$-layered graph out of $G_M$. 

\paragraph{Layered Graph Construction.} To construct the layered graph $G_L = (V_L, E_L)$ with costs $c_L : E_L \rightarrow \{0,1\}$, weights $\ell_L : E_L \rightarrow \{1, 2, \dots, \texttt{poly}(n) \}$, demand set $\mathcal{D}_L$, and distance bounds $Dist_L : \mathcal{D}_L \rightarrow \mathbb{N}_{\geq 0}$, 
we first create $\beta + 1$ copies of each vertex $u \in V$ so that $u$ corresponds to vertices $u_0, u_1, \dots, u_{\beta}$ in $V_L$.
For each edge $(u,v) \in E_M$ we add edges $\{ (u_i, v_{i+1})$ : $i \in [0, \beta-1] \}$ to $E_L$. 
For each edge $e = (u_i, v_{i+1})$ of this type, we set $\ell_{L}(e) = \ell(u,v)$. We also set $c_{L}(e) = 1$ if $(u,v) \in \widetilde{E}$; otherwise we set $c_{L}(e) = 0$.  
For each vertex in $V_L$, we also add edges $\{ (u_i, u_{i+1}) : i \in [0, \beta-1]  \}$ to $E_L$, and set their costs and weights to $0$. 
Finally, we add a demand $(s_0, t_\beta)$ to $\mathcal{D}_L$ for demand each $(s,t) \in \mathcal{D}$.

By design, $(i,j)$-junction trees in $G_M$ correspond to $(i,j)$-junction trees in $G_L$ (and vice versa) with the same density. The proof of this is straightforward, \iflong and is given in the following pair of lemmas. \fi

\iflong
\begin{lemma}
\label{cl:input_to_layered}
    For any $(i,j)$-junction tree in $G_M$, there exists an $(i,j)$-junction tree in $G_L$ with the same density. 
\end{lemma}
\begin{proof}
    Fix an $(i,j)$-junction tree $T$ in $G_M$. We build an $(i,j)$-junction tree $T_L$ in $G_L$ with the same density. Tree $T$ is the union of an in-arborescence, which we denote as $A^{in}$, and an out-arborescence, denoted by $A^{out}$, both of which are rooted at some node $r$. We add node $r_i \in V_L$ to $T_L$ as its root. 
    For each vertex $u \in A^{in}$, let $h_u$ be the number of hops (edges) on the $u-r$ path in $A^{in}$; likewise, for each vertex $w \in A^{out}$, let $h_w$ be the number of hops on the $r-w$ path in $A^{out}$. For each edge $(u,v) \in A^{in}$, we add edge $(u_{i-h_u}, v_{i-h_u+1})$ (and corresponding nodes) to $T_L$. Similarly, for each edge $(u,v) \in A^{out}$, we add edge $(u_{i+h_u}, v_{i+h_u+1})$ (and corresponding nodes) to $T_L$.
   
    Finally, we add ``dummy paths'' to $T_L$ to ensure that the corresponding demands in $\mathcal{D}_L$ are settled. These dummy paths will handle demand-settling paths in $T$ that have fewer that $\beta$ hops. Let $A^{in}_L$ and $A^{out}_L$ denote the in- and out-arborescences of $T_L$, respectively. For each vertex $v_x \in A^{in}_L$, we add edges $\{ (v_k, v_{k+1}) :  0 \leq k < x  \}$ (and corresponding nodes) to $T_L$ (if they don't already exist in $T_L$). Similarly for each vertex $v_x \in A^{out}_L$, we add edges $\{ (v_k, v_{k+1}) :  x \leq k < \beta  \}$ (and corresponding nodes) to $T_L$ (if they don't already exist in $T_L$). With these dummy paths, we ensure that if demand $(s,t)$ is settled by a path in $T$ with fewer than $\beta$ hops, then the corresponding demand $(s_0, t_\beta)$ is also settled in $T_L$.  

    Tree $T_L$ has the same cost as $T$: for every edge $(u,v) \in T$ we add an edge $(u_k, v_{k+1})$, for some integer $k$, to $T_L$. Recall that $(u,v)$ and $(u_k, v_{k+1})$ have the same cost, for any $0 \leq k < \beta$. All other edges in $T_L$ (i.e., all edges on dummy paths) are of the form $(u_k, u_{k+1})$, for some $k$, and have cost $0$. 
    
    The number of demands satisfied in both trees is also the same, which we show by mapping each demand $(s,t) \in \mathcal{D}$ settled by $T$ to a unique demand $(s_0, t_\beta) \in \mathcal{D}_L$ settled by $T_L$ (and vice versa). We now show that if $(s,t)$ is settled by $T$, then $(s_0, t_\beta)$ is settled by $T_L$. Let $P = (s, a, b, c \dots, t)$ be the path in $T$ that settles $(s,t)$. Then, path $P_L = (s_0, \dots, s_k, a_{k+1}, b_{k+2}, c_{k+3}, \dots, t_{m}, \dots , t_\beta )$ is in $T_L$, for some $k, m$. Note that subpaths $(s_0, \dots, s_k )$ and $(t_{m}, \dots , t_\beta )$ are dummy subpaths.
    Paths $P$ and $P_L$ have the same length---any edge $(u,v) \in P$ has the same length as its corresponding edge $(u_k, v_{k+1})$ (for some $k$) in $P_L$, and all other edges (i.e., edges from the dummy subpaths $(s_0, \dots, s_k)$ and $(t_m, \dots, t_\beta)$, if they exist) have length $0$. We therefore have that $d_{T_L}(s_0,t_\beta) = d_T(s,t) \leq  Dist(s,t) = Dist_L(s_0,t_\beta)$. Also, path $P_L$ has exactly $\beta$ hops. Thus, demand $(s_0, t_\beta)$ is settled by $T_L$. It is also not difficult to see that each demand $(s_0, t_\beta) \in \mathcal{D}_L$ settled by $T_L$ can be mapped to a unique demand $(s,t) \in \mathcal{D}$ settled by $T$, using similar arguments. Trees $T$ and $T_L$ therefore have the same cost and settle the same number of demands, and so have the same density. 
\end{proof}

\begin{lemma}
\label{cl:layered_to_input}
    For any $(i,j)$-junction tree in $G_L$, there exists an $(i,j)$-junction tree in $G_M$ with the same density.
\end{lemma}
\begin{proof}
     Given an $(i,j)$-junction tree $T_L$ of $G_L$, we can build an $(i,j)$-junction tree $T$ in $G_M$ with the same density. Note that edges only exist between adjacent layers in $G_L$ (and in $T_L$)---namely, all edges in $T_L$ are of the from $(u_k, v_{k+1})$ for some $k$. For each $k \in [0, \beta]$ and for each edge $(u_k, v_{k+1})$ in $T_L$ such that $u \neq v$, we add edge $(u, v) \in G_M$ (and corresponding nodes) to $T$. 
    
    Trees $T_L$ and $T$ have the same cost: For each edge $(u_k, v_{k+1}) \in T_L$ (for some $k$) with cost $1$, we add edge $(u,v)$ to $T$, which also has cost $1$. For all other edges in $T_L$ (all of which have no cost), we either add the corresponding edge to $T$, which also has no cost, or we add no edge. 
    
    Both trees also settle the same number of demands. Each demand $(s_0, t_\beta) \in \mathcal{D}_L$ settled by $T_L$ can be mapped to a unique demand $(s,t) \in \mathcal{D}$ settled by $T$. Let $P_L = (s_0, \dots, s_k, a_{k+1}, b_{k+2}, c_{k+3}, \dots, t_{m}, \dots , t_\beta )$ be the path that settles $(s_0, t_\beta)$ in $T_L$. Then, the path $P = (s, a, b, c, \dots, t)$ is in $T$. Paths $P$ and $P_L$ have the same length---any edge of the form $(u_k,v_{k+1}) \in P_L$ (for some $k$), where $u \neq v$, has the same length as its corresponding edge $(u, v) \in P$. All other edges in $P_L$ (i.e., edges from the dummy subpaths if they exist) have length $0$. We therefore have that $d_T(s,t) = d_{T_L}(s_0,t_\beta) \leq Dist_{L}(s_0,t_\beta) = Dist(s,t)$. Path $P$ also has at most $\beta$ hops, so $(s,t)$ is settled by $T$. It is also not difficult to see that each demand $(s, t) \in \mathcal{D}$ settled by $T$ can be mapped to a unique demand $(s_0,t_\beta) \in \mathcal{D}_L$ settled by $T_L$, using similar arguments. We've shown that $T_L$ and $T$ have the same cost and settle the same number of demands, and so they have the same density.
\end{proof}

Let $\Delta(T)$ denote the density of junction tree $T$. The above lemmas imply the following:

\begin{corollary}
\label{cor:equivalent}
    Let $T^*$ be the min-density $(i,j)$-junction tree (over all possible $i,j$) in $G_M$, and let $T_L^*$ be the min-density $(i,j)$-junction tree (over all possible $i,j$) in $G_L$. Then, $\Delta(T^*) = \Delta(T_L^*)$.
\end{corollary}
\begin{proof}
    By Lemma~\ref{cl:input_to_layered}, we have that $\Delta(T_L^*) \leq \Delta(T^*)$.  By Lemma~\ref{cl:layered_to_input}, $\Delta(T^*) \leq \Delta(T_L^*)$. Therefore, $\Delta(T^*) = \Delta(T_L^*)$.
\end{proof}

We now use this Corollary to reduce from finding min-density $(i,j)$-junction trees in $G_M$ to finding them in $G_L$. We say that {\jt} is the optimization problem of finding the minimum density $(i,j)$-junction tree in an input graph, over all possible values of $i,j$.

\begin{lemma}
\label{lem:reduction}
     If there is an $\alpha$-approximation algorithm for {\jt} on $G_L$, then there is also an $\alpha$-approximation algorithm for {\jt} on $G_M$. 
\end{lemma}
\begin{proof}
    Suppose we have an $\alpha$-approximation for {\jt} on graph $G_L$. Then, the following is an algorithm for {\jt} on $G_M$. First, run the $\alpha$-approximation algorithm on $G_L$, and let $T_L$ be the tree returned by the algorithm. Using the procedure described in Lemma~\ref{cl:layered_to_input}, we can build (in polynomial time) a valid $(i,j)$-junction tree $T$ of $G_M$ with the same density as $T_L$. The density of $T$ is as follows:
    \begin{align*}
        \Delta(T) = \Delta(T_L) \leq \alpha \cdot \Delta(T_L^*) = \alpha \cdot \Delta(T^*).
    \end{align*}
    The final equality is due to to Corollary~\ref{cor:equivalent}.
\end{proof}

\fi

\subsubsection{Junction Tree-Finding Subroutine}
We now show that we can find low-density junction trees at each iteration of the algorithm. Although $(i,j)$-junction trees are hopbounded by definition,  we can now use the following length-bounded junction tree-finding subroutine of~\cite{GKL23} to find hopbounded junction trees, thanks to the reduction to the $\beta$-layered graph $G_L$.

\begin{lemma}[Lemma 16 of \cite{GKL23}]
\label{lem:og_JT_alg}
    For any constant $\epsilon > 0$, there is a polynomial-time approximation algorithm for finding the minimum density distance-preserving junction tree. That is, there is a polynomial time algorithm which, given a weighted directed $n$-vertex graph $G = (V,E)$ where each edge $e \in E$ has cost $c(e) \in \mathbb{R}_{\geq 0}$ and integral length $\ell(e) \in \{0,1, \dots, \poly(n)\}$, terminal pairs $\mathcal{D} \subseteq V \times V$, and distance bounds $Dist : \mathcal{D} \rightarrow \mathbb{N}$ for each terminal pair $(s,t) \in \mathcal{D}$, approximates the following problem to within an $O(n^\epsilon)$ factor:
    \begin{itemize}
        \item Find a non-empty set of edges $F \subseteq E$ minimizing the ratio:
    \end{itemize}
    \begin{align*}
        \min_{r \in V} \frac{\sum_{e \in F} c(e)}{|\{(s,t) \in \mathcal{D} : d_{F,r}(s,t) \leq Dist(s,t) \}|}
    \end{align*}
    where $d_{F,r}(s,t)$ is the length of the shortest path using edges in $F$ which connects $s$ to $t$ while going through $r$ (if such a path exists). We call this problem {\ljt}.
\end{lemma}

This gives an $O(n^\epsilon)$-approximation algorithm for finding the min-density $(i,j)$-junction tree on $G_M$. 

\begin{lemma}
\label{lem:hop_JT_approx}
    There is an $O(n^\epsilon)$-approximation for {\jt} on $G_M$.
\end{lemma}
\iflong 
\begin{proof}
    By Lemma~\ref{lem:reduction}, we can prove the lemma by giving an $O(n^\epsilon)$-approximation for {\jt} on $G_L$. The algorithm is as follows: Simply run the algorithm of Lemma~\ref{lem:og_JT_alg} on $G_L$.

    We now show that the tree $T_L$ returned by this algorithm is an $(i,j)$-junction tree of $G_L$, where $i+j = \beta$, and that the density of $T_L$ is at most a factor $O(n^\epsilon)$ of the density of the optimal tree.
    Tree $T_L$ has some root $r_i$.
    Fix a demand $(s_0,t_\beta) \in \mathcal{D}_L$ that has length at most $Dist_L(s,t)$ in $T_L$. Due to the structure of $G_L$, the path from $s_0$ to $t_\be$ in $T_L$ must have $i$ hops from $s_0$ to $r_i$ and $\beta - i$ hops from $r_i$ to $t_\beta$. Thus $T_L$ is an $(i,j)$-junction tree. To see that this is an $O(n^\epsilon)$-approximation, first observe that the optimal density $(i,j)$-junction tree is a feasible solution to {\ljt} on $G_L$. As for the approximation ratio, the algorithm gives a $O(|V_L|^\epsilon) = O(\beta^\epsilon |V|^\epsilon) = O(n^{\epsilon'})$ approximation, where $n = |V|$ and $\epsilon' > 0$ is an arbitrarily small constant.
\end{proof}
\else
\fi

\begin{lemma}
\label{lem:hop_JT_alg}
    There is a polynomial-time algorithm that finds an $(i,j)$-junction tree with density $O(\beta n^\epsilon \cdot {{\opt}}^2 / |D|) )$, where $D$ is the set of unsettled demands in $G$.
\end{lemma}
\begin{proof}
    By Lemma~\ref{lem:existence}, there exists an $(i,j)$-junction tree with density $O(\beta \cdot {{\opt}}^2 / |D|) )$. We can run the $O(n^\epsilon)$-approximation algorithm (Lemma~\ref{lem:hop_JT_approx}) on $G_M$, which outputs an $(i,j)$-junction tree with density $O(\beta n^\epsilon \cdot {{\opt}}^2 / |D|) )$.
\end{proof}

\subsubsection{Proof of Theorem \ref{thm:junction_tree}}

By iteratively buying these low-density $(i,j)$-junction trees, we get an  $O(\beta n^\epsilon \cdot {{\opt}} )$-approximation for {\hopset}.

\begin{proof}
    The algorithm for {\hopset} builds and returns subgraph $H$, which is initialized as empty. Let $D$ be initialized as the set of unsettled demands in the input.
    The algorithm first finds an $(i,j)$-junction tree $T$ with density $O(\beta n^\epsilon \cdot {{\opt}}^2 / |D|) )$, as described in Lemma~\ref{lem:hop_JT_alg}. It then removes $T$ from $G_M$, adds $T$ to $H$, and removes all settled demands from $D$. This process repeats until all demands are settled. 
     Now we show that the algorithm gives an $O(\beta n^\epsilon \cdot {{\opt}}^2 )$-approximation. Suppose the algorithm runs for $\ell$ total iterations. For iteration $k$ of the algorithm, let $T_k$ be the $(i,j)$-junction tree found at that iteration, let $c(T_k)$ be its cost, and let $s_k$ be the number of demands settled by $T_k$. Let $D_k$ be the set of unsettled demands at the start of iteration $k$. The cost of $H$ is the following:
     \begin{align*}
         c(H) = \sum_{k=1}^{\ell} c(T_k) 
         = \sum_{k=1}^{\ell} O\left(   \frac{\beta n^\epsilon \cdot {{\opt}}^2 }{|D_k|}    \cdot s_k \right)  
         &= O \left( \beta n^\epsilon \cdot {{\opt}}^2     \sum_{k=1}^{\ell} \frac{s_k}{|D_k|} \right)  \\
         &= O \left( \beta n^\epsilon \cdot {{\opt}}^2   \cdot \log n \right) \tag{$|D|$\textsuperscript{th} Harmonic number} \\
         &= \widetilde{O}  \left( \beta n^\epsilon \cdot {{\opt}} \right) \cdot {\opt}.
     \end{align*}
    Note that $\sum_{k=1}^{\ell} \frac{s_k}{|D_k|} = O(H_{|D|}) = O(\log n)$, where $H_{|D|}$ is the $|D|$\textsuperscript{th} Harmonic number.
\end{proof}

\subsection{Star Sampling with Randomized LP Rounding Algorithm} \label{sec:star_sampling_rounding}

In this section we prove the following theorem.

\begin{theorem} \label{thm:bbmry_alg}
    There is a randomized algorithm for directed {\hopset} with expected approximation ratio $O( n \ln{n} \, \big/ \sqrt{\opt})$.
\end{theorem}

The pair of algorithms we give closely follow the $\widetilde{O}(n^{2/3})$- and $\widetilde{O}(\sqrt{n})$-approximations for the unweighted $k$-spanner problem, given by~\cite{DK11} and~\cite{BBMRY11}, respectively. The $k$-spanner algorithm is a trade-off between two algorithms: an arborescence sampling algorithm for settling a class of edges (or demands) that they call ``thick,'' and a randomized LP rounding algorithm for settling ``thin'' edges. For hopsets we will settle these thick demands by sampling directed stars instead of arborescences, and for thin demands we will use a similar LP rounding approach. Although our hopset algorithms are similar to the $k$-spanner algorithms, we get a different approximation ($\widetilde{O}( n \big/ \sqrt{\opt})$ for hopsets versus $\widetilde{O}(\sqrt{n})$ for spanners). This is because \cite{DK11, BBMRY11} take advantage of the fact that for spanners, ${\opt} \geq n-1$. This is not the case for hopsets so we get a different approximation out of the algorithms, in terms of {\opt}. This approach is also similar to that of~\cite{CDKL20} for pairwise distance preservers, where again, $\Omega(n)$ is not a lower bound for {\opt}.

We note that our $O( n \ln{n} \, \big/ \sqrt{\opt})$-approximation is achieved by trading off the two aforementioned algorithms (star-sampling and randomized-rounding). To later achieve the optimal trade-off with other algorithms, one should a priori treat each of these two algorithms as separate, with their own individual approximation ratios. It is however equivalent to trade these two algorithms off first and treat them as one combined algorithm, which we do going forward. This is because these are our only algorithms that will depend on the ``local neighborhood size'' parameter.

To define thick and thin demands, we must first define subgraphs $G^{s,t}$ for all demands $(s,t)$, as in \cite{DK11, BBMRY11}:

\begin{definition}
    For a demand $(s,t) \in \mathcal{D}$, let $G^{s,t} = (V^{s,t}, E^{s,t})$ be the subgraph of $G_M$ induced by the vertices on paths in $\mathcal{P}_{s,t}$. We call $|V^{s,t}|$ the \textit{local neighborhood size}.
\end{definition}

\begin{definition}[Thick and Thin Demands] 
    Let $b$ be a parameter in $[1,n]$. If $|V^{s,t}| \geq n/b$ then the corresponding demand $(s,t)$ is thick, otherwise it is thin. We shall always assume that $b = \sqrt{{\opt}}$.
\end{definition}

Let $\mathcal{D}_{thick}$ and $\mathcal{D}_{thin}$ be the set of all thick and thin demands, respectively. We will run two algorithms to build two edge sets, $E'$ and $E''$, such that all thick demands are settled by $E'$ and all thin demands are settled by $E''$. The set $E'$ will have cost $O(bn\ln{n})$ in expectation, while $E''$ will have cost $O((n/b)\ln{n} \cdot {\opt})$ in expectation. The optimal trade-off of these algorithms has $b = \sqrt{\opt}$, so each edge set will have cost $O(n\ln{n} \cdot \sqrt{{\opt} })$ in expectation.

\subsubsection{Star-Sampling Algorithm for Thick Demands} \label{sec:thick}
We describe the random sampling subroutine for constructing the edge set $E'$, which will settle all thick demands (Algorithm~\ref{alg:star_sample}). 

\begin{algorithm}[h]
\DontPrintSemicolon

\textbf{Input:} 
Graph $G_M = (V, E_M)$ \\

Let $E' \gets \emptyset, \; S \gets \emptyset$ \tcp*{Set $S$ is only used for the analysis} 

\ForEach{index $i = 1, 2, \dots, b\ln{n}$}{
    $v \gets$ a uniformly random element from $V$ \\
    $T^{in}_v \gets$ inward star of $G_M$ rooted at $v$ \\
    $T^{out}_v \gets$ outward star of $G_M$ rooted at $v$ \\
    $E' \gets E' \cup T^{in}_v \cup T^{out}_v , \; S \gets S \cup \{v\}$ 
} 
\ForEach{unsettled demand $(s,t) \in \mathcal{D}_{thick}$}{
    $E' \gets E' \cup (s,t)$
} 
\textbf{Return} $E'$ \;

\caption{\label{alg:star_sample} Star-Sampling Algorithm}
\end{algorithm}

This algorithm is nearly identical to that of~\cite{DK11}. The only difference is that, since we operate on the weighted transitive closure of $G$, we build directed in- and out-stars as opposed to the shortest path in- and out-arborescences used for the spanner setting. 

We now show that $E'$ has the desired cost in expectation. While~\cite{DK11} proves this for spanners, it is easy to see that a near identical argument also holds for hopsets in $G_M$. We restate the proof for completeness.

\begin{lemma}[\hspace{1sp}\cite{DK11}] \label{lem:thick}
    Algorithm~\ref{alg:star_sample}, in polynomial time, computes an edge set $E'$ that settles all thick demands and has expected cost $O(bn\ln{n} )$. If $b = \sqrt{{\opt}}$, then the expected size is $O(n\ln{n} \cdot \sqrt{{\opt}})$.
\end{lemma}
\iflong
\begin{proof}  
    After the execution of the first for loop in Algorithm~\ref{alg:star_sample}, $|E'| \leq 2(n-1)b \ln{n}$.

    If some vertex $v$ from a set $V^{s,t}$ appears in the set $S$ of vertices selected by Algorithm~\ref{alg:star_sample}, then $T^{in}_v$ and $T^{out}_v$ contain shortest, $1$-hop paths from $s$ to $v$ and from $v$ to $t$, respectively. Thus, both paths are contained in $E'$. Since $v \in V^{s,t}$, the sum of lengths of these two paths is at most $Dist(s,t)$. Therefore, if $S \cap V^{s,t} \neq \emptyset$, then the demand $(s,t)$ is settled. For a thick demand $(s,t)$, the set $S \cap V^{s,t}$ is empty with probability at most $(1-1/b)^{b\ln{n}} \leq e^{-\ln{n}} = 1/n$. Thus, the expected number of unsettled thick demands added to $E'$ in the final for loop of Algorithm~\ref{alg:star_sample} is at most $|\mathcal{D}|/n \leq n$.

    The final for loop ensures that $E'$, returned by the algorithm, settles all thick demands. Computing in- and out-stars and determining whether a demand is thick can be done in polynomial time.
\end{proof}
\else
\fi

\subsubsection{Randomized LP Rounding Algorithm for Thin Demands} 
We now give the algorithm for finding a set $E''$ to settle thin demands. \cite{BBMRY11} introduces the notion ``anti-spanners,'' which is crucial for the algorithm and analysis for settling thin demands. In particular, they formulate an anti-spanner covering LP that captures the problem of settling all thin demands. They then solve the LP (with high probability) by constructing a separation oracle that utilizes randomized rounding. We will also use randomized LP rounding, though instead of rounding the solution to an ``anti-hopset'' covering LP, we will round based on~\ref{lp:hopset}.  Our LP is stronger than the ``anti-hopset'' covering LP, since our LP is for \textit{fractional} cuts against valid paths, while the anti-hopset covering LP is only for integer cuts.

Going forward, we will assume without loss of generality that we know the value of the optimal solution---${\opt}$ is in $\{0, 1, \dots, n^2 \}$, so we can just try each of these values for ${\opt}$ and return the smallest hopset found over all tries. We can therefore replace the objective function of~\ref{lp:hopset} with the following:
\begin{align*}
    \sum_{e \in \widetilde{E}} x_e \leq {\opt} \tag{4}
\end{align*}

We use this modified version of~\ref{lp:hopset} for the randomized rounding algorithm. Given a fractional solution $\bm{\mathrm{x}}^*$ to~\ref{lp:hopset}, our algorithm will return an edge set $E''$ that, with high probability, will cost at most $2{\opt} \cdot 2(n/b)  \ln{n}$ and satisfy all thin demands (see Algorithm~\ref{alg:random_rounding}). We say that the algorithm \textit{fails} if $c(E'') > 2{\opt} \cdot 2(n/b)  \ln{n}$ or if $E''$ does not satisfy all thin demands. The algorithm will fail with low probability.

\begin{algorithm}[h]
\DontPrintSemicolon

\textbf{Input:} Graph $G_M = (V, E_M)$, \ref{lp:hopset} fractional solution $\bm{\mathrm{x}}^*$ \\

Let $E'' \gets \emptyset$ \; \;

\tcp{sample edges into $E''$}
\ForEach{edge $e \in E_M $}{
    Let $p_e \gets \min(1, 2(n/b) \ln{n} \cdot x^*_e)$ \;
    Add $e$ to $E''$ with probability $p_e$  } \;

\If{$E''$ settles all thin demands}{
    \textbf{Return} $E''$} 
    
\Else{ 
    \textbf{Return} $E_M \setminus E$  }

\caption{\label{alg:random_rounding} Randomized LP Rounding Algorithm }
\end{algorithm}

To show that with high probability, Algorithm~\ref{alg:random_rounding} does not fail, we start by defining ``anti-hopsets,'' the analogous of anti-spanners in the hopset setting. For a given demand $(s,t)$, an anti-hopset is a set of edges such that removing them from $G_M$ results in no valid paths from $s$ to $t$ in what remains.

\begin{definition}[Anti-Hopsets]
    An edge set $C \subseteq E$ is an anti-hopset for demand $(s,t) \in \mathcal{D}$ if there is no $\beta$-hopbounded path of length at most $Dist(s,t)$ in $G_M \setminus C$. If no proper subset of an anti-hopset $C$ is an anti-hopset, we say that $C$ is a minimal.
\end{definition}

Thus, a set of edges is a valid hopset if and only if it is a hitting set for the collection of (minimal) anti-hopsets---that is, to be a valid hopset, an edge set must include at least one edge from every anti-hopset.
We now show that the probability is exponentially small that the algorithm fails. The argument is very similar to that given by~\cite{BBMRY11} for spanners; we state it \iflong here \fi for completeness.

\begin{lemma}[Theorem 2.2 of \cite{BBMRY11}] \label{lem:thin_fail}
    The probability that Algorithm~\ref{alg:random_rounding} fails is exponentially small in $n$.
\end{lemma}
\iflong
\begin{proof}
    There are two different events that can cause the algorithm to fail:   
    \begin{enumerate}
        \item The cost of the sampled set $E''$ is too high---that is, $c(E'') > 2{\opt} \cdot 2(n/b)  \ln{n}$. The expected cost of $E''$ is at most $2(n/b) \ln{n} \, \cdot \sum_{e \in E_M \setminus E} x_e \leq {\opt} \cdot 2(n/b)  \ln{n}$. By the Chernoff bound (recall that $b = \sqrt{{\opt}}$), we have that $\Pr[c(E'') > 2{\opt} \cdot 2(n/b)  \ln{n} ] \leq e^{-c \cdot {\opt} \cdot (n/b) \ln{n}} = e^{-c n\ln{n} \cdot \sqrt{{\opt}}} = e^{-\Omega(n \ln(n))}$. Thus, the probability that the algorithm fails because $c(E'') > 2{\opt} \cdot 2(n/b)  \ln{n}$ is exponentially small.
        \item $E''$ does not settle all thin demands. We prove that the probability that $E''$ does not settle all thin demands (that is, that $E''$ does not intersect all minimal anti-hopsets for thin demands) is exponentially small in the following Lemma.
    \end{enumerate}
    
    \begin{lemma}[Lemma 2.3 of~\cite{BBMRY11}] \label{lem:thin_not_settle}
        The probability that there exists a demand $(s,t)$ and a minimal anti-hopset $C$ for it such that $C \subset E^{s,t} \setminus E''$ is at most $|\mathcal{D}_{thin}| \cdot \left( 1 / b n \right)^{n/b}$. In particular, if $b = \sqrt{{\opt}}$, then the probability is at most $|\mathcal{D}_{thin}| \cdot (1 \, / \, n\sqrt{{\opt}})^{n/\sqrt{\opt}}$.
    \end{lemma}
    \begin{proof}
        First, we bound the total number of minimal anti-hopsets for thin demands. 

        \begin{proposition}[Proposition 2.1 of \cite{BBMRY11}] \label{prop:thin1}
            If $(s,t)$ is a thin demand, then there are at most $(n/b)^{n/b}$ minimal anti-hopsets for $(s,t)$.
        \end{proposition}
        \begin{proof}
            Fix a thin demand $(s,t)$ and consider an arbitrary minimal anti-hopset $C$ for $(s,t)$. For the rest of this argument, for any two paths of the same length between the same pair of vertices, we consider one path to be shorter than the other if that path has fewer hops. More specifically, for any vertex pair $x, y \in V$ and any two $x-y$ paths $P, P'$ that have the same length, we say that $P$ is shorter than $P'$ if $P$ has fewer hops (edges).
            
            Let $A_C$ be the outward shortest path tree (arborescence) rooted at $s$ in the graph $(V^{s,t}, E^{s,t} \setminus C)$. Our tie-breaking of same-length paths with different hops ensures that $A_C$ includes shortest paths with the lowest number of hops (recall that every edge $(u,v) \in E_M$ is a shortest path from $u$ to $v$). 
            Denote by $f^{(\beta)}_{A_C}(u)$ the $\beta$-hopbounded distance from $s$ to $u$ in the tree $A_C$. If there is no $\beta$-hopbounded directed path from $s$ to $u$ in $A_C$, let $f^{(\beta)}_{A_C}(u) = \infty$. 
            
            We show that $C = \left\{ (u,v) \in E^{s,t} : f^{(\beta)}_{A_C}(u) + \ell(u,v) <  f^{(\beta)}_{A_C}(v) \right\}$, and thus, $A_C$ uniquely determines $C$ for a given thin demand $(s,t)$. If $(u,v) \in C$, then, since $C$ is a \textit{minimal} anti-hopset, there exists a $\beta$-hopbounded path from $s$ to $t$ of length at most $Dist(s,t)$ in the graph $(G_M \setminus C) \cup \{(u,v)\}$. This path must lie in $(G^{s,t} \setminus C) \cup \{(u,v)\}$ and must contain the edge $(u,v)$. Thus, the $\beta$-hopbounded distance from $s$ to $t$ in the graph $(G^{s,t} \setminus C) \cup \{(u,v)\}$ is at most $Dist(s,t)$ and is strictly less than $f^{(\beta)}_{A_C}(t)$. Hence, $A_C$ is not the shortest path tree in the graph $(G^{s,t} \setminus C) \cup \{(u,v)\}$. Therefore, $f^{(\beta)}_{A_C}(u) + \ell(u,v) <  f^{(\beta)}_{A_C}(v)$. If $(u,v) \in E^{s,t}$ satisfies the condition $f^{(\beta)}_{A_C}(u) + \ell(u,v) <  f^{(\beta)}_{A_C}(v)$, then $(u,v) \notin E^{s,t} \setminus C$ (otherwise, $A_C$ would not be the shortest path tree), hence $(u,v) \in C$.

            We now count the number of outward trees rooted at $s$ in $G^{s,t} \setminus C$. For every vertex $u \in V^{s,t}$, we may choose the parent vertex in at most $|V^{s,t}|$ possible ways (if a vertex is isolated we assume that it is its own parent), thus the total number of trees is at most $|V^{s,t}|^{|V^{s,t}|} \leq (n / b)^{n/b}$
        \end{proof}

        \begin{proposition}[Proposition 2.2 of \cite{BBMRY11}] \label{prop:thin2}
            For a demand $(s,t)$ and a minimal anti-hopset $C$ for $(s,t)$, the probability that $E'' \cap C = \emptyset$ is at most $e^{-2(n/b)\ln{n}}$.
        \end{proposition}
        \begin{proof}
            Suppose there is an anti-hopset edge $e \in C$ such that $x_e^* \geq (2(n/b)\ln{n})^{-1}$. In this case, $e$ is in $E''$ with probability $1$, and we are done. Otherwise, the probability that $e$ is in $E''$ is exactly $2(n/b)\ln{n} \cdot x_e$. In this case, the probability that $E''$ does not include $e$ is 
            \begin{align*}
                \prod_{e \in C} \left( 1 - 2 (n/b) \ln{n} \right) < \exp{\left(-\sum_{x \in C} 2 (n/b) \ln{n} \cdot x^*_e \right)} \leq e^{-2 (n/b) \ln{n}}.
            \end{align*}
            The first inequality holds from the fact that $1-x < \exp{(-x)}$ for $x > 0$. For the  last inequality, observe that every anti-hopset is an integer cut against valid paths. Thus, each anti-hopset $C$ corresponds to an~\ref{lp:hopset} constraint of the form $\sum_{e \in E_M} z_e x_e \geq 1$, where $\bm{\mathrm{z}} \in \mathcal{Z}_{s,t}$ is the indicator vector for cut $C$.
        \end{proof}
        
    To finish the proof of Lemma~\ref{lem:thin_not_settle}, we use Propositions~\ref{prop:thin1} and~\ref{prop:thin2} to take a union bound over all minimal anti-hopsets for all thin demands. Let $\mathcal{S}_{s,t}$ be the set of all minimal anti-hopsets for a thin demand $(s,t)$, and let $\mathcal{S}$ be the collection of all minimal anti-hopsets for all thin demands. The probability that $E''$ does not intersect all minimal anti-hopsets in $\mathcal{S}$ is the following:
    \begin{align*}
        \Pr[E'' \text{ is not a hitting set for } \mathcal{S}] &\leq  \sum_{(s,t) \in \mathcal{D}_{thin}} \sum_{C \in \mathcal{S}_{s,t}} e^{-2 (n/b) \ln{n}} \\
        &\leq  |\mathcal{D}_{thin}| \cdot \left( \frac{n}{b} \right)^{n/b} \cdot e^{-2 (n/b) \ln{n}} \\
        &= |\mathcal{D}_{thin}| \cdot \left( \frac{1}{b n} \right)^{n/b} \\
        &= |\mathcal{D}_{thin}| \cdot \left( \frac{1}{n \sqrt{{\opt}}} \right)^{\frac{n}{\sqrt{{\opt}}}} \tag{$b = \sqrt{{\opt}}$}.
    \end{align*}
    If ${\opt} = \widetilde{\Theta}(n^{2})$, then we can achieve an $\widetilde{O}(1)$-approximation by just returning $E_M \setminus E''$. Thus, we can assume without loss of generality that ${\opt} = \widetilde{o}(n^{2})$.  Given this, we have that the probability that $E''$ does not cover all minimal anti-hopsets for a thin demand $(s,t)$ is exponentially small.
    \end{proof}

\end{proof}
\else
\fi

\subsubsection{Proof of Theorem~\ref{thm:bbmry_alg}}
\begin{proof}
    All thick demands can be satisfied by running Algorithm~\ref{alg:star_sample} to build $E'$, which has expected cost $O(n \ln{n} \cdot \sqrt{{\opt}})$ (by Lemma~\ref{lem:thick}) and runs in polynomial time. The thin demands can be satisfied by running Algorithm~\ref{alg:random_rounding}, which runs in polynomial time. Algorithm~\ref{alg:random_rounding} fails with exponentially small probability (in which case we return all possible hopset edges, $\widetilde{E}$), and thus the expected cost of $E''$ is at most ${\opt} \cdot  2(n/b) \ln{n} + o(1) = O(n \ln{n} \cdot \sqrt{{\opt}} )$ (Lemma~\ref{lem:thin_fail}). Thus the overall approximation ratio is $O( n \ln{n} / \sqrt{{\opt}})$. 
\end{proof}

\subsection{Trade-Offs with Existential Bounds}\label{sec:existential}

There are a number of constructive existential results for hopsets that we trade off with our junction tree-based algorithm from Section~\ref{sec:junction_tree} (an $\widetilde{O}(\be n^\epsilon \cdot {\opt})$-approximation) and our star-sampling/randomized-rounding algorithm from Section~\ref{sec:star_sampling_rounding} (an $\widetilde{O}(n / \sqrt{\opt})$-approximation) to give approximations in several regimes. 
Our junction tree algorithm gives much better approximations than all other existential results when $\opt$ and $\beta$ are relatively small, so it will be used in the trade off for all regimes. The star-sampling/randomized-rounding algorithm gives improved approximations over the junction tree algorithm as ${\opt}$ gets larger. 

\iflong  
We first state a folklore existential bound that we will use throughout. For exact hopsets in both directed and undirected graphs, the following bound is the best-known (and is known to be tight \cite{BH23folklore}). Exact hopsets satisfy any distance bound function, so we can trade off the hopset produced by the folklore existential bound with other algorithms in any regime. For more restricted regimes---specifically for limited stretch, hopbound, and undirected graphs---we will trade off with stronger existential bounds to get improved approximations.

\begin{lemma}
    Given any weighted directed graph $G$ and a parameter $\beta >0$, there is an exact hopset of $G$ with hopbound $\beta$ and size $\widetilde{O}(n^2/\beta^2)$. This hopset can be constructed in polynomial time via random sampling.
\end{lemma}

This implies the following straightforward characterization based on $\opt$, which will allow us to trade the folklore construction off with our other approximation algorithms that also depend on $\opt$.

\begin{corollary} \label{cor:existential_folklore}
    There is a randomized polynomial-time $\widetilde{O}(n^2/ (\beta^2 \cdot \opt))$-approximation for directed {\hopset}.
\end{corollary}

As $\be$ gets larger, this folklore construction gives improved approximations over the junction tree and star-sampling/randomized-rounding algorithms, especially when $\opt$ is also relatively large.

\subsubsection{Directed Hopsets with Arbitrary Distance Bounds}
We trade off our junction-tree and star-sampling/randomized-rounding algorithms with the Corollary~\ref{cor:existential_folklore} folklore approximation to achieve an overall $\widetilde{O}(n^{4/5 + \epsilon})$-approximation for directed {\hopset}. 

\begin{theorem} \label{thm:main_result}
    There is a randomized  polynomial-time $\widetilde{O}(n^{4/5 + \epsilon})$-approximation for directed {\hopset}.
\end{theorem}
\begin{proof}
    We trade off the junction tree based algorithm from Section~\ref{sec:junction_tree}, the star-sampling/randomized-rounding algorithms from Section~\ref{sec:star_sampling_rounding}, and the algorithm of Corollary~\ref{cor:existential_folklore} to give an algorithm for directed {\hopset}; that is, we run all three algorithms and return the solution with the minimum cost. The overall approximation ratio of this algorithm, $\al$, is the maximum over all possible $\beta, \opt$ of the minimum cost:
    \begin{align*}
        \al &= \widetilde{O}\left( \max_{\be, {\opt}} \left( \min_{} \left\{ \be n^\epsilon \cdot {\opt}, \mbox{ } \frac{n}{\sqrt{{\opt}}}, \mbox{ }\frac{n^2}{\be^2 \cdot {\opt}}  \right\} \right) \right).                    
    \end{align*}
    
    The value of $\al$ is $\widetilde{O}(n^{4/5+\epsilon})$, achieved at $\be = \widetilde{O}(n^{2/5})$, ${\opt} = \widetilde{O}(n^{2/5 - \epsilon})$ (this is also the point at which all three curves intersect). Additionally, each of the three algorithms runs in polynomial-time, so the overall algorithm is polynomial-time.
\end{proof}

When $\beta = \widetilde{O}(n^{2/5})$, the folklore construction gives worse approximations than the trade off between our other algorithms. By just trading off the junction tree and star-sampling/randomized-rounding algorithms, we achieve a better approximation in this regime.

\begin{theorem} \label{thm:small_be_dir_gen}
    When $\be = \widetilde{O}(n^{2/5})$, there is a randomized polynomial-time $\widetilde{O}(\be^{1/3} \cdot n^{2/3 + \epsilon})$-approximation for directed {\hopset}.
\end{theorem}
\begin{corollary}
    When $\be = \widetilde{O}(1)$, Theorem~\ref{thm:small_be_dir_gen} gives an $\widetilde{O}(n^{2/3 + \epsilon})$-approximation for {\hopset}.
\end{corollary}

We can also get improved approximations in the large $\beta$ setting by trading off the junction tree algorithm with the Corollary~\ref{cor:existential_folklore} folklore approximation. Note that because the folklore construction has a better inverse dependence on $\beta$ than the star-sampling/randomized-rounding algorithm (in fact, the latter has \textit{no} dependence on $\be$), the folklore construction performs better than star-sampling/randomized rounding when $\be$ is sufficiently large.

\begin{theorem} \label{thm:big_be_dir_gen}
    When $\be = \widetilde{\Omega}(n^{2/5})$, there is a randomized polynomial-time $\widetilde{O}( n^{1+\epsilon} / \sqrt{\vphantom{\be^k} \be})$-approximation for directed {\hopset}.
\end{theorem}

\subsubsection{Directed Hopsets with Small Stretch}

For $(1+\epsilon)$-approximate directed hopsets, the best known existential bound is the following:

\begin{lemma}[Theorem 1.1 of~\cite{BW23}]
    For any directed graph $G$ with integer weights in $[1,W]$, given $\epsilon \in (0,1)$ and $\beta \geq  20\log n$, there is a $(1+\epsilon)$-hopset with hopbound $\beta$ and size:
    \begin{itemize}
        \item $\widetilde{O}\left(\frac{n^2\log^2(nW)}{\beta^3 \epsilon^2  }\right)$ for $\beta \leq n^{1/3}$,
        \item $\widetilde{O}\left(\frac{n^{\frac{3}{2}} \log^2(nW)}{\beta^{3/2} \epsilon^2}\right)$ for $\beta >n^{1/3}$.
    \end{itemize} 
    This hopset can be constructed in polynomial time.
\end{lemma}
\begin{corollary} \label{cor:existential_W}
    When $\beta \geq 20\log n$, there is a polynomial-time approximation for directed $(1+\epsilon)$ {\hopset} with approximation ratio:
        \begin{itemize}
        \item $\widetilde{O}\left(\frac{n^2\log^2(nW)}{\beta^3 \epsilon^2 \cdot \opt  }\right)$ for $\beta \leq n^{1/3}$,
        \item $\widetilde{O}\left(\frac{n^{\frac{3}{2}} \log^2(nW)}{\beta^{3/2} \epsilon^2 \cdot \opt}\right)$ for $\beta >n^{1/3}$
    \end{itemize}
    where $\epsilon \in (0,1)$.
\end{corollary}

Using the Corollary~\ref{cor:existential_W} algorithm, we get an improved approximation for directed {\hopset} when we restrict to $(1+\epsilon)$ stretch. 

\begin{theorem} \label{thm:dir_eps}
    When $\beta \geq 20\log n$ and $\epsilon \in (0,1)$, there is a randomized polynomial-time $\widetilde{O}(n^{3/4 + \epsilon'} \cdot \epsilon^{-\frac{1}{4}})$-approximation for directed stretch-$(1+\epsilon)$ {\hopset}, where $\epsilon' > 0$ is an arbitrarily small constant.
\end{theorem}
\begin{proof}
    The approximation is achieved by trading off the junction tree algorithm of Section~\ref{sec:junction_tree}, the star-sampling/randomized-rounding algorithms of Section~\ref{sec:star_sampling_rounding}, the Corollary~\ref{cor:existential_folklore} folklore algorithm, and the Corollary~\ref{cor:existential_W} algorithm. Note that $W = poly(n)$, so the $\log(nW)$ factor is hidden by the $\widetilde{O}(\cdot)$ notation.
\end{proof}

\subsubsection{Undirected Hopsets with Small Stretch}
For \textit{undirected} hopsets with $(1+\epsilon)$ stretch, there is the following constructive existential result.

\begin{lemma}[\hspace{1sp}\cite{elkin2019RNC}] \label{lem:existential_undir_eps}
    For any weighted undirected graph $G$, any integer $\eta \geq 1$, and any $0 < \rho < 1$, there is a randomized algorithm that runs in polynomial time in expectation and computes a hopset with size $O(n^{1 + 1/\eta})$, which is a hopset for any $\epsilon \in (0,1)$ with hopbound
    \begin{align*}
        \beta = \left( \frac{\log(\eta) + 1/\rho }{\epsilon}   \right)^{\log(\eta) + \frac{1}{\rho}+1} .
    \end{align*}
\end{lemma}

Let $W_0(x)$ be the principle branch of the Lambert $W$ function. When $x \geq 3$, the function is upper bounded by $\ln{x} - (1/2) \ln{\ln{x}}$. The Lemma~\ref{lem:existential_undir_eps} existential result implies the following:

\begin{corollary} \label{cor:existential_undir_eps}
    Let $\eta = \lfloor \beta^{1/W_0(\ln{\beta})} \rfloor > \be^{1/(\ln{\ln{\be}}-\frac{1}{2}\ln{\ln{\ln{\be}}})}$ (inequality holds when $\be \geq 3$). There is a randomized polynomial-time $O(n^{1 + 1/\eta} / {\opt})$-approximation for undirected stretch-$(1+\epsilon)$ {\hopset}, where $\epsilon \in (0,1)$.
\end{corollary}

For some insight into the behavior of $\eta$, note first that for all $\beta \geq 3$, $\eta \geq 6$. Additionally, the $\eta$ function grows faster than $\ln{\be}$, but much slower than $\be$. The Corollary~\ref{cor:existential_undir_eps} construction performs better than the star-sampling/randomized-rounding algorithms as $\opt$ grows, resulting in improved approximations compared to the directed graph, arbitrary stretch regime when $\be$ is relatively small.

\begin{theorem} \label{thm:undir_eps}    
    Let $\eta = \lfloor \beta^{1/W_0(\ln{\beta})} \rfloor > \be^{1/(\ln{\ln{\be}}-\frac{1}{2}\ln{\ln{\ln{\be}}})}$ (inequality holds when $\be \geq 3$). When $\beta = \widetilde{O}(n^{\frac{1}{2} - \frac{1}{2\eta}})$, there is a randomized polynomial-time $\widetilde{O}(\sqrt{\be} \cdot n^{\frac{1}{2} + \frac{1}{2\eta} + \epsilon'})$-approximation for undirected $(1+\epsilon)$-stretch {\hopset}, where $\epsilon \in (0,1)$, and $\epsilon' > 0$ is an arbitrarily small constant.
\end{theorem}
\begin{proof}
    The approximation is achieved by trading off the junction tree algorithm of Section~\ref{sec:junction_tree} with the Corollary~\ref{cor:existential_undir_eps} algorithm.
\end{proof}

\subsubsection{Undirected Hopsets with Odd Stretch}

The following result is directly implied by Thorup-Zwick approximate distance oracles.

\begin{lemma}[\hspace{1sp}\cite{TZ05}]
    Let $k \geq 1$ be an integer. For any weighted undirected graph $G$, a stretch-$(2k-1)$ hopset with hopbound $2$ and size $O(kn^{1+1/k})$ can be found in polynomial-time.
\end{lemma}
\begin{corollary} \label{cor:existential_undir_k}
    Let $k \geq 1$ be an integer. There is a polynomial-time $O(kn^{1+1/k} / {\opt})$-approximation for undirected stretch-$(2k-1)$ {\hopset}.
\end{corollary}

We can use Corollary~\ref{cor:existential_undir_k} to get an improved approximation (over the directed graph, arbitrary stretch setting) for undirected stretch-$(2k-1)$ {\hopset} when $\be$ is relatively small.

\begin{theorem} \label{thm:undir_gen_stretch}
    Let $k \geq 1$ be an integer. When $\be = \widetilde{O}(k^{-1/2} \cdot \, n^{\frac{1}{2} - \frac{1}{2k}}) $, there is a polynomial-time $\widetilde{O}(\sqrt{k \be \vphantom{\be^k}} \, \cdot \, n^{\frac{1}{2} + \frac{1}{2k} + \epsilon})$-approximation for undirected stretch-$(2k-1)$ {\hopset}, where $\epsilon > 0$ is an arbitrarily small constant.
\end{theorem}
\begin{proof}
    The approximation is achieved by trading off the junction tree algorithm of Section~\ref{sec:junction_tree} and the Corollary~\ref{cor:existential_undir_k} algorithm.
\end{proof}

\fi


\iflong
\section{Label Cover Hardness for Directed Hopsets and Shortcut Sets} \label{app:hardness}
In this section we prove Theorem~\ref{thm:LB-main}. 
 In particular, we prove hardness of approximation for directed shortcut sets on directed graphs, which immediately implies hardness for directed hopsets for any stretch bound as long as the hopbound is at least $3$.  We use a reduction from Min-Rep (the natural minimization version of \LabelCover) in a very similar way to the hardness proofs for graph spanners from~\cite{Kor01,EP07,DKR16,CD16}.  The main technical difficulty / difference is that these previous reductions for spanners heavily use the fact that only edges from the original graph can be included in the spanner.  This is not true of shortcut sets (or hopsets), and makes it simpler to reason about the space of ``all feasible spanners'' than it is to reason about the space of ``all possible shortcut sets''.  

The Min-Rep problem, first introduced by~\cite{Kor01} and since used to prove hardness of approximation for many network design problems, can be thought of as a minimization version of \LabelCover with the additional property that the alphabets are represented explicitly as vertices in a graph.  A Min-Rep instance consists of an undirected bipartite graph $G = (A, B, E)$, together with partitions of $A_1, A_2, \dots, A_{m}$ of $A$ and $B_1, B_2, \dots, B_m$ of $B$ with the property that $|A_i| = |A_j| = |B_i| = |B_j$ for every $i,j \in [m]$\footnote{In some versions of the problem we do not assume that each partition has the same number of parts or that each $|A_i|=|B_i|$, but those assumptions can both be made without loss of generality so we do so for convenience}.  Each part of one of the partitions is called a \emph{group}.  If there is at least one edge between $A_i$ and $B_j$, then we say that $(i,j)$ is a \emph{superedge}.  Our goal is to choose a set $S \subseteq A \cup B$ of vertices of $V$ if minimum size so that, for every superedge $(i,j)$, there is some vertex $x \in A_i \cap S$ and $y \in B_j \cap S$ with $\{x,y\} \in E$.  In other words, we must choose vertices to \emph{cover} each superedge.  Any such cover is called a REP-cover, and our goal is to find the minimum size REP-cover.  Note that we must choose at least one vertex from every group (assuming that each group is involved in at least one superedge), and hence $OPT \geq 2m$.  

Since Min-Rep is essentially \LabelCover, which is a two-query PCP, the PCP theorem implies hardness of approximation for Min-Rep.  More formally, the following hardness is known~\cite{Kor01}.

\begin{theorem}[\cite{Kor01}] \label{thm:MinRep-hardness}
    Assuming that $NP \not\subseteq DTIME(2^{polylog(n)})$, then for any constant $\epsilon > 0$ there is no polynomial-time algorithm that approximates Min-Rep better than $2^{\log^{1-\epsilon} n}$.  In particular, no polynomial time algorithm can distinguish between when $OPT = 2m$ and when $OPT \geq 2m \cdot 2^{\log^{1-\epsilon} n}$
\end{theorem}

\subsection{Reduction from Min-Rep to shortcut sets}
We can now give our formal reduction from Min-Rep to shortcut sets with hopbound $\beta \geq 3$.  Consider some instance of Min-Rep $G = (A, B, E)$ with partition $A_1, A_2, \dots, A_m$ and $B_1, B_2, \dots, B_m$.  First we create two nodes for every group: let 
\begin{align*}
&A' = \{a_1^1, a_1^2, a_2^1, a_2^2, \dots, a_m^1, a_m^2\} & \text{and} &  &B' = \{b_1^1, b_1^2, b_2^1,b_2^2, \dots, b_m^1, b_m^2\}.
\end{align*}
Then for every edge $e = \{a, b\}$, we create a set of $\beta - 3$ vertices $V_e = \{v_e^1, v_e^2, \dots, v_e^{\beta-3}\}$ (note that these sets are empty if $\beta=3$).  Our final vertex set is 
\begin{align*}
    V' = V \cup A' \cup B' \cup (\cup_{e \in E} V_e).
\end{align*}

We now create the (directed) edges of the graph.  For each $e = \{a,b\} \in E$ we create a path:
\begin{align*}
    P_e &= \{(a, v_e^1)\} \cup \{(v_e^i, v_e^{i+1}) : i \in [\beta-4]\} \cup \{(b, v_e^{\beta-3})\}.
\end{align*}
If $\beta = 3$, we simply set $P_e = \{(a,b)\}$.  We can now specify our final edge set:
\begin{align*}
    E' = &\left(\cup_{e \in E} P_e\right) \\
    &\cup \{(a_i^1, a_i^2) : i \in [m]\} \cup \{(b_i^2, b_i^1) : i \in [m]\} \\
    &\cup \{(a_i^2, x) : i \in [m], x \in A_i\} \cup \{(x, b_i^2) : i \in [m], x \in B_i\}.
\end{align*}

Our final shortcut set instance is $G' = (V', E')$ with hopbound $\beta$.

\subsection{Analysis}
The analysis of this reduction has two parts, completeness and soundness (so-called due to their connections back to probabilistically checkable proofs).  For completeness, we show that if the original Min-Rep instance has $OPT = \gamma$ then there is a shortcut set of size at most $\gamma$.  For soundness, we show that if there is a shortcut set of size $\gamma$, then there is a REP-cover of size at most $O(\gamma)$.  Together with Theorem~\ref{thm:MinRep-hardness}, these will imply our desired hardness bound for shortcut sets.  To get Theorem~\ref{thm:LB-main}, one just needs to observe that in $G'$, if there is a path from some node $u$ to some node $v$ then in fact \emph{every} path from $u$ to $v$ has exactly the same length.  Hence any shortcut set is a hopset with the same hopbound and stretch $1$, and any hopset with stretch $\alpha$ is also shortcut set.

\subsubsection{Completeness}
We begin with completeness, which is the easier direction.  We prove the following lemma.

\begin{lemma} \label{lem:completeness}
    If there is a REP-cover of the Min-Rep instance with size $\gamma$, then there is a shortcut set for $G'$ with hopbound $\beta$ and size $\gamma$.
\end{lemma}
\begin{proof}
Suppose that there is a REP-cover $S \subseteq V$ with $|S|= \gamma$.  Consider the edge set 
\begin{align*}
    S' &= \{(a_i^1, x) : i \in [m], x \in S \cap A_i\} \cup \{(x, b_i^1) : i \in [m], x \in S \cap B_i\}.
\end{align*}
In other words, we add an edge between each ``outer'' node representing a group and the nodes in the REP-cover that are in that group (directed oppositely for the left and right sides of the bipartite graph).  Clearly $|S'| = |S| = \gamma$, so it remains to show that $S'$ is a shortcut set with hopbound $\beta$.  

An easy observation is that in $G'$, the only nodes $x, y$ where $y$ is reachable from $x$ but there are more than $\beta$ hops on every $x \leadsto y$ path are when $x = a_i^1$ and $y = b_j^1$ for some $i,j \in [m]$ where $(i,j)$ is a superedge.  So consider such an $a_i^1, b_j^1$.  Then since $(i,j)$ is a superedge and $S$ is a REP-cover, we know that there is some $a,b \in S$ with $a \in A_i$ and $b \in B_j$ and $\{a,b\} \in E$.  Then by definition $S'$ includes the edges $(a_i^1, a)$ and $(b, b_j^1)$.  Thus there is a path from $a_i^1$ to $b_j^1$ which first uses the edge to $a$ added by $S$, then uses the edges of $P_{\{a,b\}}$ to get to $b$, and then uses the edge to $b_j^1$ added by $S$.  This path has exactly $\beta$ hops by construction.  Hence $S'$ is a shortcut set with hopbound $\beta$.
\end{proof}

\subsubsection{Soundness}
We now argue about soundness.  To do this, we first need the following definition.

\begin{definition} \label{def:canonical}
    A shortcut edge is called \emph{canonical} if it is of the form $(a_i^1, a)$ for some $i \in [m]$ and $a \in A_i$ or of the form $(b, b_i^1)$ for some $i \in [m]$ and $b \in B_i$.  A shortcut set $S'$ is called \emph{canonical} if every edge in $S'$ is canonical. 
\end{definition}

We now show that without loss of generality, we may assume that any shortcut set is canonical.

\begin{lemma} \label{lem:canonical}
    Suppose that $S$ is a shortcut set of $G'$ with hopbound $\beta$ and size $\gamma$.  Then there is a canonical shortcut set $S'$ of $G'$ with hopbound $\beta$ and size at most $2 \gamma$.
\end{lemma}
\begin{proof}
    Consider some $(x,y) \in S$ which is not canonical.  Since it is not canonical but also must exist in the transitive closure of $G'$, there must exist a super edge $(i,j)$ so that $x$ and $y$ are as follows.
    \begin{itemize}
        \item $x$ is either $a_i^1, a_i^2$, is in $A_i$, or is in $P_e$ for some $e = \{a,b\}$ with $a \in A_i$ and $b \in B_j$.
        \item $y$ is either $b_j^1, b_j^2$, is in $B_j$, or is in $P_e$ for some $e = \{a,b\}$ with $a \in A_i$ and $b \in B_j$.
    \end{itemize}
    In particular, this implies that $(x,y)$ is only useful in decreasing the number of hops to at most $\beta$ for the pair $(a_i^1, b_j^1)$.  In other words, if we \emph{removed} $(x,y)$ from $S$, then the only pair of nodes which be reachable but not within $\beta$ hops would be $(a_i^1, b_j^1)$.  So choose some arbitrary edge $\{a,b\} \in E$ with $a \in A_i$ and $b \in B_j$, and let $(a_i^1, a)$ and $(b, b_j^1)$ be the \emph{replacement edges} for $(x,y)$.  Note that these replacement edges are both canonical, and suffice to decrease the number of hops from $a_i^1$ to $b_j^1$ to $\beta$.

    So we let $S'$ consist of all of the canonical edges of $S$ together with the two replacement edges for every noncanonical edge in $S$.  Clearly $|S'| \leq 2|S| =2\gamma$, and $S'$ is a shortcut set with hopbound $\beta$.  
\end{proof}

With this lemma in hand, it is now relatively straighforward to prove soundness.

\begin{lemma} \label{lem:soundness}
    If there is a shortcut set for $G'$ with hopbound $\beta$ and size $\gamma$, then there is a REP-cover of the Min-Rep instance with size at most $2\gamma$.
\end{lemma}
\begin{proof}
    Let $\hat S$ be a shortcut set for $G'$ with hopbound $\beta$ and size $\gamma$.  By Lemma~\ref{lem:canonical}, there is a shortcut set $S'$ for $G'$ with hopbound $\beta$ and size at most $2\gamma$.  Let $S$ be the set of nodes such that they are an endpoint of some edge in $S'$ and are in $A \cup B$.  In other words, for every canonical edge $(a_i^1, a) \in S'$ we add $a$ to $S$, and similarly for every canonical edge $(b, b_i^1) \in S'$ we add $b$ to $S$.  Clearly $|S| \leq |S'| \leq 2\gamma$, so we just need to prove that $S$ is a REP-cover.

    Consider some superedge $(i,j)$.  Then $G'$, we know that $a_i^1$ can reach $b_j^1$ but every such path has length $\beta+2$.  Since $S'$ is a canonical shortcut set, this means that it must add an edge $(a_i^1, a)$ and an edge $(b, b_j^1)$ for some $a \in A_i$ and $b \in B_i$ with $(a,b) \in E$.  Hence $a,b \in S$ by definition.  Thus $S$ is a REP-cover as claimed.  
\end{proof}
\else 
\fi

\bibliographystyle{alpha}
\bibliography{refs}


\end{document}